\documentclass[aps,showpacs,prl,twocolumn,superscriptaddress,nobibnotes]{revtex4-1}

\usepackage{exscale}
\usepackage{bbm}
\usepackage{graphicx}
\usepackage{latexsym}
\usepackage{times}
\usepackage[T1]{fontenc}
\usepackage{enumerate}
\usepackage{bbold}
\usepackage{color}
\usepackage{mathtools}
\usepackage[defaultcolor=red]{changes}
\usepackage[normalem]{ulem}
\usepackage{amsfonts,amsmath,amssymb,amsthm,dsfont}
\usepackage[colorlinks=true,citecolor=blue,urlcolor=blue]{hyperref}
\usepackage{natbib}
\usepackage{lineno}
\usepackage{physics}

\theoremstyle{plain}
\newtheorem{thm}{Theorem}
\newtheorem{lem}{Lemma}
\newtheorem{fact}{Fact}

\newcommand{\proj}[1]{\ket{#1}\!\!\bra{#1}}
 
\def\NN{\mathbbm{N}}

\begin{document}

\title{No graph state is preparable in quantum networks with bipartite sources and no classical communication}
\author{Owidiusz Makuta}
\email{owidiusz.makuta@gmail.com}
\affiliation{Center for Theoretical Physics, Polish Academy of Sciences, Aleja Lotnik\'{o}w 32/46, 02-668 Warsaw, Poland}

\author{Laurens T. Ligthart}
\affiliation{Institute for Theoretical Physics, University of Cologne, Zülpicher Straße 77, D-50937 Cologne, Germany}

\author{Remigiusz Augusiak}
\affiliation{Center for Theoretical Physics, Polish Academy of Sciences, Aleja Lotnik\'{o}w 32/46, 02-668 Warsaw, Poland}

\begin{abstract}
In research concerning quantum networks, it is often assumed that the parties can classically communicate with each other. However, classical communication might introduce a substantial delay to the network, especially if it is large. As the latency of a network is one of its most important characteristics, it is interesting to consider quantum networks in which parties cannot communicate classically and ask what limitations this assumption imposes on the possibility of preparing multipartite states in such networks. We show that graph states of an arbitrary prime local dimension known for their numerous applications in quantum information cannot be generated in a quantum network in which parties are connected via sources of bipartite quantum states and the classical communication is replaced by some pre-shared classical correlations. We then generalise our result to arbitrary quantum states that are sufficiently close to graph states. 
\end{abstract}

\maketitle

\section*{Introduction} Out of all proposed real-life implementations of quantum information, quantum networks stand as one of the most promising ones. We already have many ideas for possible applications of quantum networks such as quantum key distribution \cite{Lo_2014,Liao_2017,PhysRevLett.113.190501}, clock synchronisation \cite{K_m_r_2014}, parallel computing \cite{9651415} or even quantum internet \cite{Kimble2008-dy, doi:10.1126/science.aam9288, 8910635, ILLIANO2022109092}. What is more, their relative simplicity as compared with other quantum technologies, makes them perhaps the closest ones to commercial implementation. This sentiment is also supported by the significant progress in the experimental implementation of quantum networks that has recently been made \cite{Wang_2010,Sasaki_2011,Fr_hlich_2013,PhysRevLett.120.030501,Yin2020-uh}.

At their core, quantum networks are simply collections of parties and of sources of multipartite quantum states. Their most natural model (referred to as LOCC networks) is one that allows the parties to act with the most general local operations on their shares of the distributed states and to coordinate their actions by using classical communication. However, while connected LOCC networks enable preparing any multipartite state, the use of classical communication might be problematic for their commercial implementations. 

When considering the possible future applications of quantum networks one has to take into account that the distances between parties will substantially increase as compared to the current state-of-the-art experiments. Therefore, while classical communication between parties can be considered almost instantaneous in a lab setup, this will not be the case for quantum networks spanning many different countries or even continents. Depending on the number of rounds of measurement and communication that have to be performed for a generation of a given state, the delay caused by the travel time of the information could become a substantial factor to the run time of any procedure done via a quantum network. From that point of view, it would be beneficial, e.g. for quantum key distribution protocols, to use quantum states that require as little classical communication as possible to be generated. This has not only the potential to decrease the latency of quantum networks but also to reduce the noise therein, as the longer a quantum state has to be stored, the noisier the state gets.

One is thus forced to consider quantum networks in which 
the amount of classical communication between the parties is limited or even no communication is allowed (see Ref. \cite{SpeeKraft}). 
A possible model of quantum networks that fulfils this requirement is one in which the parties can apply arbitrary quantum channels to their 
particles, however, they cannot use classical communication. Instead, they are allowed to orchestrate their actions by using some pre-shared classical information. We call such networks LOSR (local operations and shared randomness) quantum networks. While the above no-communication assumption severely reduces the capability of generating multipartite states in LOSR quantum networks as compared to the LOCC ones, they are still more general than those in which the parties can only implement unitary operations and no randomness is shared. While the LOSR networks have become an object of intensive studies
\cite{mao2022test,PhysRevLett.125.240505,coiteux2021any,hansenne2021quantum,Luo}, it remains unclear what multipartite entangled states can actually be prepared in them.

In quantum information graph states stand as one of the most representative classes of multipartite states, including the Greenberger-Horne-Zeilinger \cite{Greenberger1989}, cluster \cite{ClusterStates} or the absolutely maximally entangled \cite{helwig2013absolutely} states. Moreover, they are key resources for many 
applications, just to mention quantum computing \cite{PhysRevLett.86.5188,Briegel_2009,Yao2012-wq}, multipartite secret sharing \cite{SecretSharing}, or quantum metrology \cite{T_th_2014}. It has thus been a question of utmost importance whether they can be prepared in quantum networks. While in a LOCC quantum network one can simply generate the graph state locally and then distribute it using quantum teleportation \cite{Werner_2001, Rigolin_2005}, this strategy cannot be applied in the LOSR case. Therefore, the question remains of whether one can generate graph states in LOSR quantum networks.

Here we answer the above question to the negative and show that no graph states of arbitrary prime local dimension (or any state sufficiently close to them) can be prepared in LOSR networks with bipartite sources. We thus generalise the recent results that the three-partite GHZ states \cite{PhysRevLett.125.240505,Kraft_2021} or any $N$-qubit graph states with $N$ up to $12$ cannot be prepared in such networks \cite{hansenne2021quantum}. Simultaneously, our work is complementary to Ref. \cite{Luo} showing that no qubit or qutrit graph states of an arbitrary number of parties can be prepared in LOSR networks even with $(N-1)$-partite sources. Our proof employs the quantum inflation method \cite{WolfeSpekkensFritz+2019,Wolfe_2021}
which is perfectly suited to tackle these types of questions \cite{PhysRevLett.125.240505,Wolfe_2021}.

\section*{Results}
\subsection*{Preliminaries}
\textit{(1) Graph states.} 
Consider a multigraph $G$ which is a graph in which any two vertices can be connected by more than one edge, but no edge can connect a vertex to itself. 
Let $\Gamma_{i,j}$ denote the number of edges connecting vertices $i$ and $j$, and let $\mathcal{N}_{i}$ be the neighbourhood of vertex $i$---the set of vertices that are connected to $i$ by at least one edge (see Fig. \ref{fig_graph} for an example). To associate a quantum state to an $N$-vertex multigraph $G$, 
we consider a Hilbert space $\mathcal{H}=\mathbb{C}_{d}^{\otimes N}$, where each qudit space $\mathbb{C}_d$ corresponds to one of the vertices of $G$; we assume
that $d$ is prime and fulfils $d-1 \geqslant \max_{i,j} \Gamma_{i,j}$. 
To each vertex $i$ we associate the following operator
\begin{equation}\label{eq_gi_def}
g_{i}= X_{i}\prod_{j\in \mathcal{N}_{i}}\left(Z_{j}\right)^{\Gamma_{i,j}},
\end{equation}
where $X$ and $Z$ are generalised Pauli matrices, 
\begin{equation}
X=\sum_{i=0}^{d-1}\ket{i+1}\!\!\bra{i}, \qquad Z=\sum_{i=0}^{d-1}\omega^{i}\proj{i},
\end{equation}
with $\omega=\exp (2\pi \mathrm{i}/d)$ and $\ket{d}\equiv\ket{0}$. The subscripts in (\ref{eq_gi_def}) label the subsystems on which these operators act. 
One defines a graph state $\ket{G}$ associated to $G$ to be the unique state in $\mathcal{H}$ obeying  $g_{i}\ket{G}=\ket{G}$ $(i=1,\ldots,N)$ (for a review see \cite{hein2006entanglement}).

\textit{(2) Quantum networks.}
Let us consider a scenario in which $N$ parties, labelled $1,\dots, N$, receive quantum states distributed by independent sources. Each party $i$ can perform an arbitrary local operation represented by a quantum channel $\mathcal{E}_{i}$, on their shares of these states. We also assume that parties cannot communicate with each other, yet they all have access to some shared randomness, which is a random variable $\lambda$ with a distribution $p_{\lambda}$. These assumptions describe a scenario called LOSR. 

\begin{figure}
	\includegraphics[width=\columnwidth]{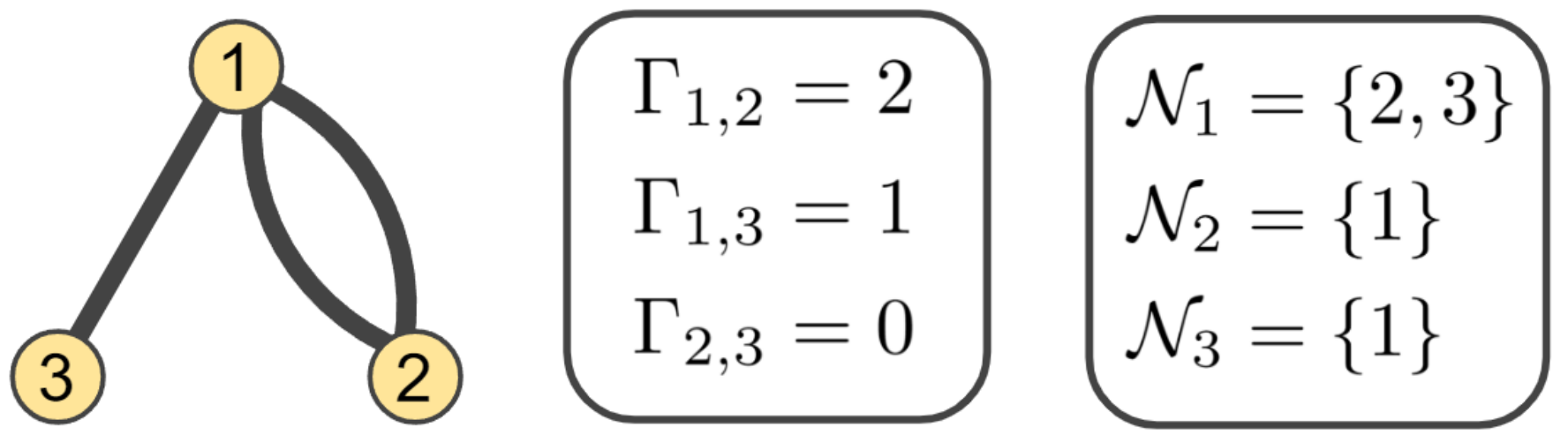}
	\caption{\textbf{Exemplary multigraph.} A an example of multigraph with three vertices.}\label{fig_graph}
\end{figure}

There is one more assumption, independent of LOSR, that we make: the sources distributing quantum states are bipartite, i.e., a single source distributes a quantum state to two parties. We say that two parties are connected if they share a state.

The most general state that can be produced in an LOSR network is given by \cite{hansenne2021quantum}
\begin{equation}\label{eq_rho}
\rho = \sum_{\lambda}p_{\lambda} \mathcal{E}_{1}^{\lambda}\otimes\dots\otimes\mathcal{E}_{N}^{\lambda}\left[\sigma_{1,2}\otimes\sigma_{1,3}\otimes \dots\otimes \sigma_{N-1,N} \right],
\end{equation}
where $\sigma_{i,j}$ denotes a state distributed between  parties $i$ and $j$, $\sum_{\lambda}p_{\lambda}=1$, and the subscript $\lambda$ denotes the dependence of local operations on the shared random variable. 

Here, three remarks are in order. First, tensor products of $\mathcal{E}_{i}^{\lambda}$ and of $\sigma_{i,j}$ are taken with respect to different sets of subsystems; while the former is taken with respect to different parties, the latter separates states from different sources.

Second, we should expect the distributed states $\sigma_{i,j}$ to depend on $\lambda$ because the sources can be classically correlated as well. However, since we do not impose any restriction on the dimension of $\sigma_{i,j}$, one can get rid of this dependency by considering a Hilbert space of sufficiently high dimension \cite{PhysRevLett.125.240505}. 

Third, in this work, we will assume that every network we consider (not including inflations) is fully-connected, i.e., each party shares a bipartite state with every other party. We can make this assumption without loss of generality because the behaviour of each quantum network can always be simulated by the fully connected one: taking $\sigma_{i,j}$ to be the maximally mixed state produces the same outcomes as removing the connection between the nodes $i$ and $j$. In other words, if $\rho$ cannot be generated in an $N$-partite LOSR network with bipartite sources that is fully connected, then it cannot be generated in any such network which is not fully connected.

\textit{(3) Network inflation.} Let us briefly describe here the network inflation method \cite{WolfeSpekkensFritz+2019,Wolfe_2021} which we use to derive our main result. Given some network $\mathcal{O}$, the inflation $\mathcal{I}$ of $\mathcal{O}$ is a network that consists of multiple copies of parties and sources from the original network. Whether two parties are connected in $\mathcal{I}$ is up to us, the only restriction is that they are connected via a copy of a source from $\mathcal{O}$.

This construction is very general as many different inflations can be considered for
a given network $\mathcal{O}$. However, here we focus on a certain class of inflations that are tailored to our proof. 
Consider an $N$-partite network $\mathcal{O}$ that we want to analyse; as already explained, we assume it to be fully connected.
In our approach, every inflation $\mathcal{I}$ of $\mathcal{O}$ consists of two copies of the parties from $\mathcal{O}$ labelled $1,\dots, N$ and $1',\dots,N'$. We assume that parties $i$ and $i'$ apply the same local operation as the original one:  $\mathcal{E}^{\mathcal{I}}_{i}=\mathcal{E}^{\mathcal{I}}_{i'} =\mathcal{E}^{\mathcal{O}}_{i}$, where the superscripts indicate the network. We also assume that each party $i$ in $\mathcal{I}$ is connected to either $j$ or $j'$ but never to both, and that two copies of the same party $i$ and $i'$ are never connected to each other. Furthermore, if two copies of parties share a state, this state is a copy of a state shared between original parties in $\mathcal{O}$. These last two assumptions imply that for every pair of parties $i,j\in\mathcal{O}$ and any inflation $\mathcal{I}$, exactly one of the following statements is true:
\begin{equation}\label{eq_sigma_1}
\sigma_{i,j}^{\mathcal{I}}=\sigma_{i',j'}^{\mathcal{I}}=\sigma_{i,j}^{\mathcal{O}}, \quad \textrm{or} \quad \sigma_{i,j'}^{\mathcal{I}}=\sigma_{i',j}^{\mathcal{I}}=\sigma_{i,j}^{\mathcal{O}}.
\end{equation}
We finally assume that in $\mathcal{I}$ the shared randomness is distributed between all copies of parties, meaning that the state generated in $\mathcal{I}$ is described by Eq. \eqref{eq_rho}.

The above assumptions allow us to establish very useful relations between expected values $\langle \cdot\rangle_{\mathcal{I}_{1}}$ and $\langle \cdot\rangle_{\mathcal{I}_{2}}$ calculated over states from two different inflations $\mathcal{I}_{1}$, $\mathcal{I}_{2}$. To this end, let us introduce an isomorphism $\chi:\mathcal{I}_{1}\rightarrow\mathcal{I}_{2}$ with an associated set $S_{\chi}\subset \{1,\dots, N\}$, that acts by swapping labels of parties $i$ and $i'$ for all $i\in S_{\chi}$. If an operator $M$ is a $2N$-fold tensor product, then we use the notation $\chi(M)$ for a swap operation between parties $i$ and $i'$ for all $i\in S_{\chi}$. With this we formulate the following fact.
\begin{fact}\label{fact_=_2}\label{fact_=_1}
Consider a network $\mathcal{O}$ and two different inflations of it, $\mathcal{I}_{1}$ and $\mathcal{I}_{2}$. Consider also two matrices $B=\bigotimes_{i\in \mathcal{I}_{1}}B_i$ and $C=\bigotimes_{i\in \mathcal{I}_{2}}C_{i}$ that act nontrivially on some subnetworks $\mathcal{I}_i'\subseteq\mathcal{I}_i$. Then, $\langle B\rangle_{\mathcal{I}_{1}}=\langle C\rangle_{\mathcal{I}_{2}}$ if there exists an isomorphism $\chi$ such that
\begin{equation}
\chi(\mathcal{I}'_1)=\mathcal{I}'_2 \quad \textrm{and}\quad \chi(B)=C.
\end{equation}
\end{fact}

The above fact can be proven by using the decomposition of the state in LOSR network \eqref{eq_rho}, taking out the sum over $\lambda$ out of the trace, and then tracing out the parties where $B$ and $C$ act trivially.

\subsection*{Main results} Let us now move on to the main results of our work, namely that no graph state can be generated in quantum networks with bipartite sources. We begin by presenting the key ingredients of our approach, which is inspired by the recent work  \cite{hansenne2021quantum}. The main idea of the proof is to show that the assumption that a graph state can be generated in a network leads to the violation of a certain inequality that follows from the lemma below (see Supplementary Note \hyperref[app_sqr]{1} for a proof and Supplementary Note \hyperref[app_bound_Ad]{2} for an alternative approach).

\begin{lem}\label{lem_ineq}
Consider two unitary matrices $A_{1},A_{2}$ acting on some Hilbert space $\mathbb{C}_D$ with $D$ being a multiple of some prime number $d\geq 2$. Assume moreover that $A_i$ satisfy $A_{1}^{d}=A_{2}^{d}=\mathbb{1}$ and $A_{1}A_{2}=\omega^{q}A_{2}A_{1}$ for some $q\in \{1,\dots,d-1\}$. Then
\begin{equation}\label{eq_ineq_sq}
\sum_{k=0}^{d-1}\langle A_{1}^{k}+ A_{2}^{k}\rangle\leqslant d+\sqrt{d},
\end{equation}
where $\langle\cdot\rangle\equiv\Tr[\rho(\cdot)]$ and the above holds true for any state $\rho$ acting on $\mathbb{C}_D$.
\end{lem}
In order to show a violation of \eqref{eq_ineq_sq} we also need the following fact.
\begin{fact}\label{fact_stab}
Consider three mutually commuting unitary matrices $B_i$ that obey $B_{i}^{d}=\mathbb{1}$. If $\langle B_{1}B_{3}\rangle = \langle B_{2}B_{3}^{\dagger}\rangle = 1$, then $\langle B_{1}B_{2}\rangle=1$.
\end{fact}
This fact follows from an observation that
$B_1B_3$ and $B_2B_3^{\dagger}$ are unitary and therefore the fact 
that $\langle B_{1}B_{3}\rangle = \langle B_{2}B_{3}^{\dagger}\rangle=1$ holds true for some
$\ket{\psi}$ implies that $B_{1}B_{3}\ket{\psi}=\ket{\psi}$ and $B_{2}B_{3}^{\dagger}\ket{\psi}=\ket{\psi}$. Since 
$B_i$ mutually commute, one concludes that $B_1B_2\ket{\psi}=\ket{\psi}$ which gives the
above implication.

Facts \ref{fact_=_1} and \ref{fact_stab} together with Lemma \ref{lem_ineq} are the key elements of our approach. However, before showing how they are combined into proof, let us present an illustrative example (for an extensive explanation of the inflation technique see Supplementary Note \hyperref[app:inflation]{3}). Let us fix $d=3$ and consider the triangle network $\mathcal{O}^{\Delta}$ presented in Fig. ~\ref{fig_tri} and the graph state $\ket{G_{\Delta}}$ corresponding to the triangle graph shown in Fig. \ref{fig_graph}, i.e., one that satisfies
$g_{i}\ket{G_{\Delta}}=\ket{G_{\Delta}}$ for 
\begin{equation}
g_{1}=X_{1}Z_{2}^{2}Z_{3},\qquad g_{2}=Z_{1}^{2}X_{2}, \qquad g_{3}=Z_{1}X_{3}.
\end{equation}
In what follows we prove that this state cannot be generated in $\mathcal{O}^{\Delta}$. Our proof is by contradiction.

Let us consider an inflation of $\mathcal{O}^{\Delta}$, denoted $\mathcal{I}_{0}^{\Delta}$ [cf. Fig. \ref{fig_tri}], and two operators defined on it, $g_1$, acting only on non-primed parties, and $Z_1^2X_{2'}$. Due to the fact that they overlap only on the non-primed party 1, it follows that $g_1 Z_1^2X_{2'}=\omega Z_1^2X_{2'} g_1$, and thus Lemma \ref{lem_ineq} implies that
\begin{equation}\label{eq_ineq_sq_ex}
\sum_{k=0}^{2}\left(\left\langle g_{1}^{k}\right\rangle_{\mathcal{I}_{0}^{\Delta}}+\left\langle (Z_{1}^{2}X_{2'})^{k}\right\rangle_{\mathcal{I}_{0}^{\Delta}}\right)\leqslant 3 + \sqrt{3},
\end{equation}
where $\langle \cdot\rangle_{\mathcal{I}_{0}^{\Delta}}$ denotes the expected value calculated on any state that can be generated in $\mathcal{I}_{0}^{\Delta}$. 
Our goal is to show that the assumption that the graph state $\ket{G}_{\Delta}$ can be generated in $\mathcal{O}^{\Delta}$ leads to violation of inequality (\ref{eq_ineq_sq_ex}). We achieve this by proving that, under the above assumption, every expected value in \eqref{eq_ineq_sq_ex} equals one.

\begin{figure}
	\centering
	\includegraphics[width=\columnwidth]{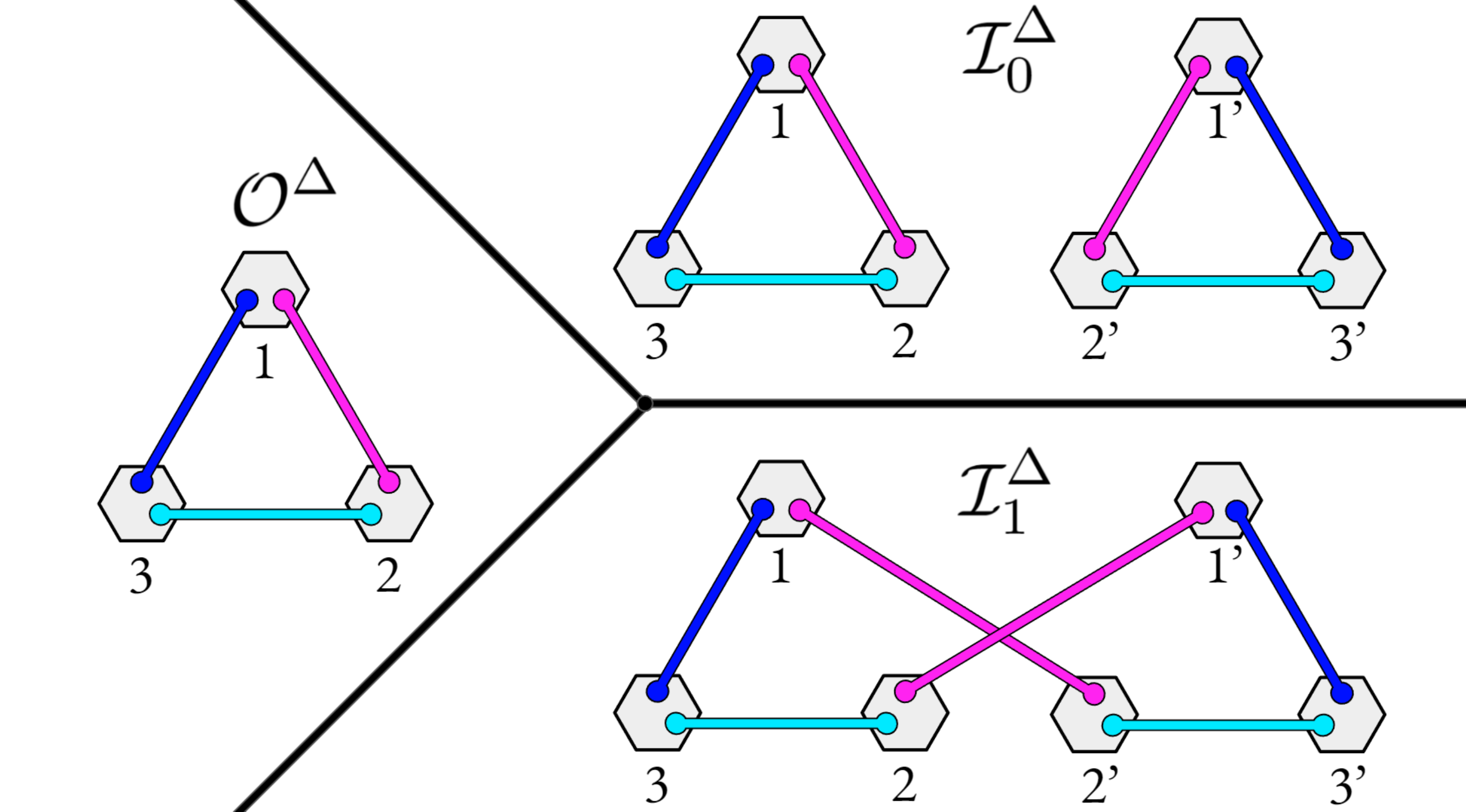}
	\caption{\textbf{Examplary network with its inflations.} A fully connected tripartite quantum network $\mathcal{O}^{\Delta}$ with its two inflations $\mathcal{I}_{0}^{\Delta}$ and $\mathcal{I}_{1}^{\Delta}$. The edges represent bipartite states shared by the parties.}\label{fig_tri}
\end{figure}

Let us focus on the first expected value in (\ref{eq_ineq_sq_ex}). Given that $\ket{G_{\Delta}}$ can be prepared in $\mathcal{O}^{\Delta}$, by tracing out all primed parties we get $\langle g_{1}^{k}\rangle_{\mathcal{I}_{0}^{\Delta}}=\langle g_{1}^{k}\rangle_{\mathcal{O}^{\Delta}}=1$. As for the second expected value in \eqref{eq_ineq_sq_ex}, we can use Fact \ref{fact_=_2} to show that
\begin{equation}\label{eq_i0i1_tri}
\left\langle Z_{1}^{2}X_{2'}\right\rangle_{\mathcal{I}_{0}^{\Delta}}=\left\langle g_{2}\right\rangle_{\mathcal{I}_{1}^{\Delta}},
\end{equation}
where $\mathcal{I}_{1}^{\Delta}$ is another inflation of $\mathcal{O}^{\Delta}$ presented in Fig. \ref{fig_tri}. 

Let us now prove that $\langle
g_{2}\rangle_{\mathcal{I}_{1}^{\Delta}}=1$. Since
the nodes $1$ and $2$ are disconnected in $\mathcal{I}_1^{\Delta}$ we cannot
directly obtain this expected value from the original network
$\mathcal{O}^{\Delta}$ and the state $\ket{G_{\Delta}}$;
still, we can compute it indirectly by employing Fact \ref{fact_stab}.
To this aim, we first notice that by tracing out all primed parties we get implies that $\langle g_{2}g_{3}\rangle_{\mathcal{I}_{1}^{\Delta}}=\langle g_{2}g_{3}\rangle_{\mathcal{O}^{\Delta}}=1$ and $\langle g_{3}^{2}\rangle_{\mathcal{I}_{1}^{\Delta}}=\langle g_{3}^{2}\rangle_{\mathcal{O}^{\Delta}}=1$.
Since the three generators $g_i$ mutually commute, one concludes from Fact \ref{fact_stab} that $\langle g_{2}g_{3}\rangle_{\mathcal{I}_{1}^{\Delta}}=1=\langle g_{3}^{2}\rangle_{\mathcal{I}_{1}^{\Delta}}$ implies
$\left\langle g_{2}\right\rangle_{\mathcal{I}_{1}^{\Delta}}=1$,
which is what we wanted to obtain. 

Hence, Eq. \eqref{eq_i0i1_tri} implies that $\left\langle Z_{1}^{2}X_{2'}\right\rangle_{\mathcal{I}_{0}^{\Delta}}=1.$
By the same argument we have $\left\langle Z_{1}X_{2'}^{2}\right\rangle_{\mathcal{I}_{0}^{\Delta}}=1,$ which implies that the left side of \eqref{eq_ineq_sq_ex} equals six, leading to a contradiction. Thus, the triangle graph state cannot be prepared in the network $\mathcal{O}^{\Delta}$.

We are now ready to present our main result that no graph states
of arbitrary local prime dimension can be produced in LOSR quantum networks
with bipartite sources, generalizing the results of 
Refs. \cite{hansenne2021quantum,Luo}.
\begin{thm}\label{thm_bip}
Consider a graph $G$ with $N\geqslant 3$ vertices and where at least one vertex $i$ has a neighbourhood  $|\mathcal{N}_{i}|\geqslant 2$. The graph state $\ket{G}\in \mathbb{C}_{d}^{\otimes N}$, where $d$ is prime, corresponding to a graph $G$ cannot be generated in an LOSR $N$-partite quantum network with bipartite sources.
\end{thm}

\begin{proof}The proof is highly technical, and so we present it in Supplementary Note \hyperref[app_proof]{4}. 
Here we describe its key ideas. 

First, in every graph satisfying the assumptions of Theorem \ref{thm_bip} we can relabel vertices so that $|\mathcal{N}_{1}\setminus \mathcal{N}_{2}|\geqslant 2$ and $\Gamma_{1,2}\neq 0$. Next, we utilise a graph transformation called local complementation (see Refs. \cite{PhysRevA.69.022316,bahramgiri2007graph}) to divide the set of all graph states into four, distinct subsets. Since the proof for three of those classes is relatively simple, here we discuss only the most complicated case of graphs for which $\mathcal{N}_{1}\cap \mathcal{N}_{2}\neq \emptyset$ and 
\begin{equation}\label{eq_n2nn_main}
\mathcal{N}_{2}\setminus \{n\} \neq \mathcal{N}_{n}\setminus \{2\}
\end{equation}
for all $n\in \mathcal{N}_{1}\cap \mathcal{N}_{2}$.
Let us consider a graph state $\ket{G}$ corresponding to a multigraph $G$ that fulfils the above assumptions. To prove that this state cannot be generated in an $N$-partite quantum network $\mathcal{O}$ we use the inflation method.
Specifically, we consider a series of inflations $\mathcal{I}_k$ (see Fig.~\ref{fig_I_k}) of the initial network $\mathcal{O}$, all having the same structure. A non-primed party $i$ (for $i\neq 2$) is connected to every other non-primed party $j$ (for $j\neq 2$) and to either $2$ or $2'$. By $R_k$ and $T_k$ we denote the sets of non-primed parties connected to party $2$ and $2'$, respectively. Likewise, we assume that every primed party $i'$ (for $i'\neq2'$) is connected to every other primed party $j'$ (for $j'\neq 2'$) and to either $2$ or $2'$; we denote the corresponding sets of parties $T_{k}'$ and $R_{k}'$.
However, the exact form of each $\mathcal{I}_{k}$ depends on the considered graph state, and so we define it later in the proof.

Having defined $\mathcal{I}_k$, we can now go on with the proof, which, as in the above example, is by contradiction. Let us consider two operators: $g_{2}$ given in Eq. \eqref{eq_gi_def} and
\begin{equation}
\tilde{g}_{1}=X_{1}Z_{2'}^{\Gamma_{1,2}}\prod_{k\in \mathcal{N}_{1}\setminus\{2\}}Z_{k}^{\Gamma_{1,k}}.
\end{equation}
While $g_{2}$ stabilises $\ket{G}$, $\tilde{g}_{1}$ is like $g_1$, however, with $Z^{\Gamma_{1,2}}$
acting on the party $2'$ instead of $2$. This difference between $g_{1}$ and $\tilde{g}_{1}$ produces a commutation relation: $\tilde{g}_{1} g_{2}=\omega^{-\Gamma_{1,2}} g_{2}\tilde{g}_{1}$. Since we also assume $\Gamma_{1,2}\neq 0$, Lemma \ref{lem_ineq} yields
\begin{equation}\label{eq_square_main}
\sum_{k=0}^{d-1}\left\langle \tilde{g}_{1}^{k} + g_{2}^{k} \right\rangle_{\mathcal{I}_{0}} \leqslant d+ \sqrt{d},
\end{equation}
where $\mathcal{I}_{0}$ is defined by the set $T_{0}=(\mathcal{N}_{1}\setminus(\mathcal{N}_{2}\cup \{2\})) \cup \{1\}$.

\begin{figure}
	\centering
	\includegraphics[width=\columnwidth]{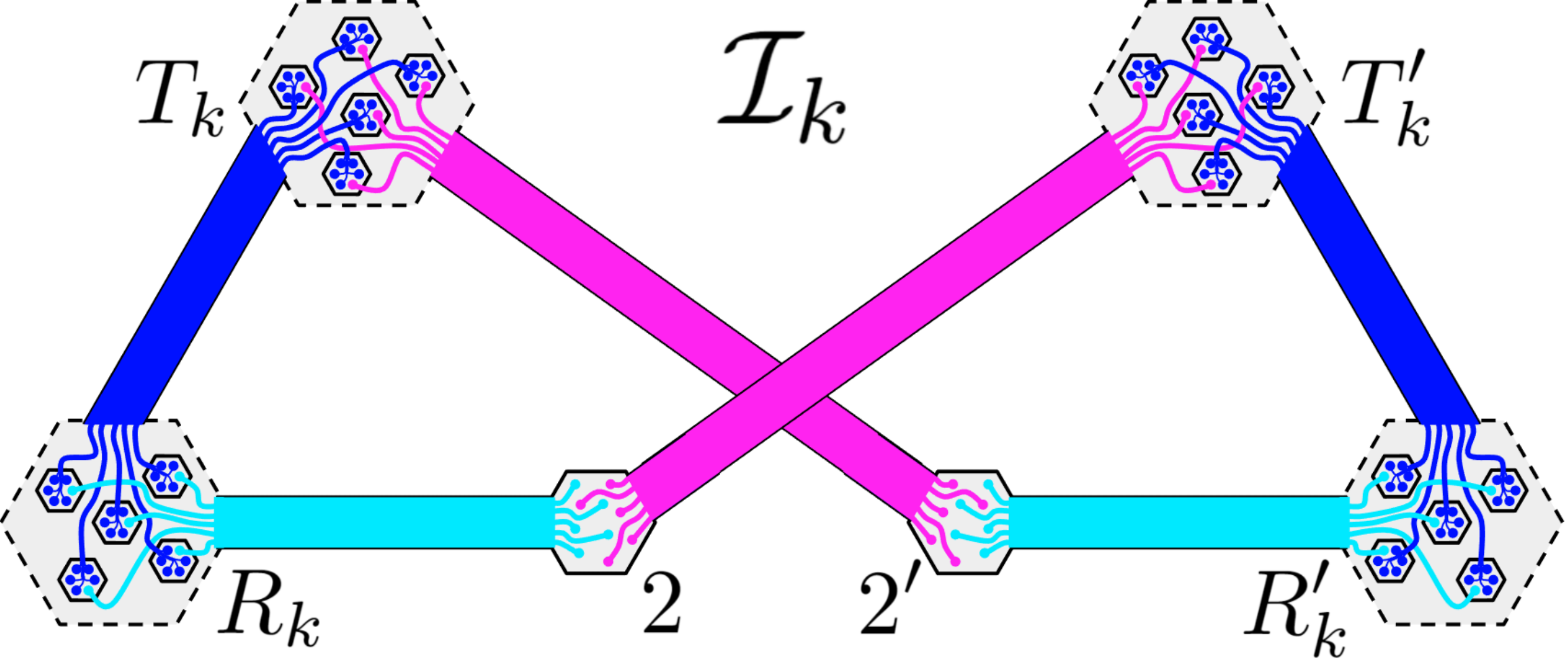}
	\caption{\textbf{Inflations $\mathcal{I}_{k}$}. Here $T_{k}$, $R_{k}$ and $T_{k}'$, $R_{k}'$ are sets of non-primed and primed parties respectively. Every party from a set is connected to all other parties from that set and if two sets are connected then every party from one set is connected to every party from the other set.
	}\label{fig_I_k}
\end{figure}

The remainder of the proof consists in showing that the assumption that a graph state $\ket{G}$ can be generated in $\mathcal{O}$ leads to violation of inequality \eqref{eq_square_main}. First, using Fact \ref{fact_=_2} we show that $\left\langle\tilde{g}_{1}\right\rangle_{\mathcal{I}_{0}} = \langle g_{1} \rangle_{\mathcal{I}_{1}}$, where $\mathcal{I}_{1}$ is another inflation (cf. Fig.~\ref{fig_I_k}) defined by $T_{1}=\mathcal{N}_{1}\cap\mathcal{N}_{2}$. We then leverage Facts \ref{fact_=_2} and \ref{fact_stab} to show that 
$\langle g_{1} \rangle_{\mathcal{I}_{2}}=1$ implies 
$\langle g_{1} \rangle_{\mathcal{I}_{1}}=1$, 
where $\mathcal{I}_{2}$ is an inflation (Fig.~\ref{fig_I_k}) defined by $T_{2}=T_{1}\setminus \{n\}$ for some $n\in T_{1}$. 
Crucially, we can perform this procedure again for $n'\in T_{2}$ which yields 
\begin{equation}
\langle g_{1} \rangle_{\mathcal{I}_{3}}=1 \quad \Rightarrow \quad \langle g_{1} \rangle_{\mathcal{I}_{2}}=1 \quad \Rightarrow \quad \langle g_{1} \rangle_{\mathcal{I}_{1}}=1,
\end{equation}
where $\mathcal{I}_{3}$ is an inflation (Fig.~\ref{fig_I_k}) defined by $T_{3}=T_{1}\setminus \{n,n'\}$. 

Repeating this procedure $q=|\mathcal{N}_{1}\cap \mathcal{N}_{2}|+1$ times produces a chain of implications
\begin{equation}\label{eq_chain}
\langle g_{1}\rangle_{\mathcal{I}_{q}}=1 \;\Rightarrow\;  \langle g_{1}\rangle_{\mathcal{I}_{q-1}}=1 \;\Rightarrow\dots\Rightarrow\; \langle g_{1}\rangle_{\mathcal{I}_{1}}=1.
\end{equation}
This is the main idea of our proof: we start from inflation $\mathcal{I}_{1}$ and we gradually make it more and more similar to the original network. In order to see how this is done notice that if $T_{q}=\emptyset$ for some $q$, then $\mathcal{I}_{q}$ is an inflation consisting of two copies of the original network. Therefore, since $g_{1}$ acts nontrivially only on non-primed parties, by tracing out all primed parties we get $\langle g_{1}\rangle_{\mathcal{I}_{q}}=\langle g_{1}\rangle_{\mathcal{O}}=1$, if we generated $\ket{G}$ in $\mathcal{O}$. This implies $\left\langle \tilde{g}_{1}\right\rangle_{\mathcal{I}_{0}} =1$, and so $\tilde{g}_1$ is a stabilizing operator of a state generated in $\mathcal{I}_{0}$, meaning that $\left\langle \tilde{g}_{1}^{k}\right\rangle_{\mathcal{I}_{0}} =1$ for all $k$. 

One can show that $g_{2}$ is a stabilizing operator of 
the state generated in $\mathcal{I}_{0}$, therefore $\langle g_{2}^{k}\rangle =1$ for all $k$. Consequently, the left side of \eqref{eq_square_main} equals $2d$, leading to a contradiction.
\end{proof}

Using the standard continuity argument one can extend the above result to any state that is sufficiently close to a graph state. Indeed, denoting $F(\rho, \ket{\psi})= \bra{\psi}\rho\ket{\psi}$, we can formulate the following theorem (see Supplementary Note \hyperref[app_fid]{5} for a proof).
\begin{thm}\label{thm_fid}
Let us consider a state $\rho$ and a graph state $\ket{G}\in \mathbb{C}_{d}^{\otimes N}$, where $d$ is prime. Moreover, let $q=|\mathcal{N}_{1}\cap \mathcal{N}_{2}|+1$ for graphs $G$ that fulfil \eqref{eq_n2nn_main} and $q=1$ in other cases. If
\begin{equation}\label{eq_fid_main}
F(\rho, \ket{G}) >1- \frac{1}{8}\left( \beta^{2}+2\gamma - \beta \sqrt{\beta^{2}+4\gamma} \right),
\end{equation}
where $\beta=2q-1$ and $\gamma = (d-\sqrt{d})/(d-1)$, then $\rho$ cannot be generated in a LOSR network with bipartite sources.
\end{thm}
Clearly, increasing $q$ increases the value of the expression on the right side of the inequality. Therefore, to maximize the number of states covered by the above theorem, it is beneficial to take the smallest $q$ possible. As an example, let us consider the graph state corresponding to the graph presented on Fig. \ref{fig_graph}. Here $q=1$, $d=3$, and so \eqref{eq_fid_main} simplifies to $F(\rho, \ket{G})>0.952$. Interestingly, the bound \eqref{eq_fid_main} can be relaxed by increasing $d$; in the limit $d\rightarrow \infty$ for $q=1$ we obtain $F(\rho, \ket{G})>0.905$.

\section*{Discussion} We showed classical communication between parties is necessary for the generation of qudit graph states of prime local dimension (and all states that are in their vicinity) in quantum networks with bipartite sources. We achieve this goal by employing the quantum inflation method. In fact, we demonstrate that the use of many different inflations of the same network might be beneficial over using just two inflations as done before in the literature, and hence our work might inspire future research involving quantum networks. Our results serve as a guide to experimental physicists who wish to implement protocols on quantum networks that involve graph states. On the other hand, they hint at a possible gain from the construction of protocols based on other states than the graph states.

Still, many questions concerning LOSR networks remain unexplored. The most obvious one is whether graph states are preparable with $k$-partite sources for $k\geqslant 3$. Even if from the application viewpoint this case seems less important than that of $k=2$, answering this question would allow us to understand quantum networks on a deeper level. One can also ask whether other classes of multipartite states can be obtained in LOSR networks. Apart from the graph states, this question was answered negatively for symmetric or antisymmetric states of any local dimension \cite{hansenne2021quantum} and pure genuinely entangled states of local dimension $2$ and $3$ in Ref. \cite{Luo}. On the other hand, it would be interesting to determine the minimal amount of classical communication required to generate graph states in LOCC networks and to identify other classes of states that are efficiently preparable in this sense.

Lastly, let us note here that while finishing this manuscript, we became aware of a related work by Y.-X. Wang \textit{et al.} \cite{related}, where the analogous statement to our Theorem \ref{thm_bip} has been derived for all qubit, and a subclass of qudit graph states. While the proof in \cite{related} covers only a subset of graph states in the case $d>2$, this includes some qudit graph states where the local dimension is not a prime number, which are not covered by our Theorem \ref{thm_bip}. Moreover, the inequality from \cite{related} used to formulate the proof (analog of Lemma \ref{lem_ineq}), allows the authors to derive a fidelity bound for the said subset of graph states that performs significantly better for small $d$ than the one we provide in Theorem \ref{thm_fid}.

\section*{Acknowledgements} We thank Felipe Montealegre-Mora, David Wierichs, and David Gross for useful discussions.
This work was supported by the National Science Center (Poland) through the SONATA BIS project (grant no. 2019/34/E/ST2/00369)
and by Germany's Excellence Strategy -- Cluster of Excellence Matter and Light for Quantum Computing (ML4Q) EXC 2004/1 -- 390534769.

\appendix
\section{Supplementary Note 1: Proof of Lemma 1}\label{app_sqr}

Here we prove Lemma 1 stated in the main text.
\setcounter{lem}{0}
\begin{lem}\label{app:lem_ineq}
Consider two unitary matrices $A_{1},A_{2}$ acting on some finite dimensional Hilbert space $\mathbbm{C}^D$ with $D$ being a multiple of some prime number $d\geq 2$. Assume moreover that $A_i$ are such that $A_{1}^{d}=A_{2}^{d}=\mathbb{1}$ and $A_{1}A_{2}=\omega^{q}A_{2}A_{1}$ for some $q\in \{1,\dots,d-1\}$. Then, for any mixed state $\rho$ acting on $\mathbbm{C}^D$, the following inequality holds true
\begin{equation}\label{eq_ineq_sq_app}
\sum_{k=0}^{d-1}\langle A_{1}^{k}+ A_{2}^{k}\rangle_{\rho}\leqslant d+\sqrt{d},
\end{equation}
where we denoted $\langle\cdot\rangle_{\rho}:=\Tr[\rho (\cdot)]$.
\end{lem}
\begin{proof}
We start from rewriting the sum in Eq. (\ref{eq_ineq_sq_app}) in the following form
\begin{equation}\label{eq_sum_app_1}
\sum_{k=0}^{d-1}\langle A_{1}^{k}+ A_{2}^{k}\rangle_{\rho}=\sum_{k=0}^{d-1}
\operatorname{tr}\left[ (A_{1}^{k}+ A_{2}^{k}) \rho\right],
\end{equation}
where $\rho$ is an arbitrary state acting on $\mathbbm{C}^D$.
Let us then assume that $A_{1}A_{2}=\omega^{q}A_{2}A_{1}$ for some $q\in\{1,\dots,d-1\}$ and that $A_{1}^{d}=A_{2}^{d}=\mathbb{1}$. As proven in Ref. \cite{Kaniewski_2019} (see Proposition B.1 in Supplementary Materials B) these two conditions imply the existence of a unitary $U:\mathbbm{C}^D\to \mathbbm{C}^d\otimes \mathcal{H}'$ such that
\begin{equation}
U A_{1}U^{\dagger} = X \otimes \mathbbm{1}, \quad U A_{2}U^{\dagger} = Z^{q} \otimes \mathbbm{1},
\end{equation}
where $\mathbbm{1}$ is the identity acting on $\mathcal{H}'$.
Denoting $\rho' = U\rho U^{\dagger}$, Eq. \eqref{eq_sum_app_1} can be rewritten as
\begin{equation}
\sum_{k=0}^{d-1}\operatorname{tr}\left[  (X^{k} + Z^{qk})\otimes \mathbb{1} \cdot \rho'\right].
\end{equation}
We can trace out the subsystems corresponding to the Hilbert space $\mathcal{H}'$ which leads us to 
\begin{equation}\label{eq_sum_app_1_2}
\sum_{k=0}^{d-1}\operatorname{tr}\left[  (X^{k} + Z^{qk}) \tilde{\rho}\right]= \sum_{k=0}^{d-1}\operatorname{tr}\left[  (X^{k} + Z^{qk}) \tilde{\rho}\right],
\end{equation}
where $\tilde{\rho}=\operatorname{tr}_{\mathcal{H}'}\rho'$. 

Let us now consider the eigendecompositions of both $X$ and $Z$ operators, 
\begin{equation}
    X=\sum_{i=0}^{d-1}\omega^i\proj{\phi_i},\qquad 
    Z=\sum_{i=0}^{d-1}\omega^i\proj{i},
\end{equation}
where $\ket{\psi_i}$ and $\ket{i}$ are the eigenvectors 
of $X$ and $Z$, respectively. Exploiting the well-known fact that
\begin{equation}
\sum_{k=0}^{d-1}\omega^{k i}=d\delta_{i,0},    
\end{equation}
it is not difficult to observe that 
\begin{equation}
\sum_{k=0}^{d-1}X^{k}=d\proj{\phi_0}, \qquad \sum_{k=0}^{d-1}Z^{kq}=d\proj{0},
\end{equation}
where $\ket{\phi_{0}}$ and $\ket{0}$ are eigenstates corresponding to 
the eigenvalue $1$ of $X$ and $Z^{q}$, respectively. This, taking into account, Eq. \eqref{eq_sum_app_1_2}, allows us to rewrite Eq. (\ref{eq_sum_app_1}) as
\begin{equation}
\sum_{k=0}^{d-1}\langle A_{1}^{k}+ A_{2}^{k}\rangle_{\rho}=d\operatorname{tr}\left[ (\proj{\phi_{0}} + \proj{0}) \tilde{\rho}\right].
\end{equation}
Let us now notice that for any pair of normalized vectors $\ket{\psi_i}$ $(i=1,2)$, the sum of projectors $P=\proj{\psi_1}+\proj{\psi_2}$ is a rank-two matrix 
whose eigenvalues are $1\pm |\langle \psi_1|\psi_2\rangle|$ and therefore 
$\operatorname{tr}[P\sigma]\leq 1+|\langle \psi_1|\psi_2\rangle|$ for any $\sigma$.
Taking into account that the eigenbasis of $X$ and $Z^{q}$ are mutually unbiased, meaning that $|\langle 0|\phi_0\rangle|=1/\sqrt{d}$, the above bound implies that 
\begin{equation}
    \sum_{k=0}^{d-1}\langle A_{1}^{k}+ A_{2}^{k}\rangle_{\rho}\leq d+\sqrt{d},
\end{equation}
which completes the proof.
\end{proof}

\section{Supplementary Note 2: An alternative equality}\label{app_bound_Ad}
While the above inequality is enough to prove the main result of this paper, we have found yet another way allowing to reach the same result that exploits an equality. Let us present it here.

\begin{lem} \label{lemma:Ad}
Given two positive real numbers $\lambda_1, \lambda_2$ and two matrices $A_{1},A_{2}$ such that $A_{1}^{d}=A_{2}^{d}=\mathbb{1}$ and $A_{1}A_{2}=\omega^{q}A_{2}A_{1}$ for $q\in \{1,\dots,d-1\}$, if $q$ and $d$ are coprime, then it holds that 
\begin{align}\label{eq_lemma_ad}
    \expval{(\lambda_1 A_1 + \lambda_2 A_2)^{nd}}_{\rho} = (\lambda_1^d + \lambda_2^d)^n \qquad \forall n\in \NN,
\end{align}
where $\rho$ is an arbitrary state acting on $\mathbbm{C}^D$.
\end{lem}
\begin{proof}
Exploiting the commutation relation $A_2 A_1=\eta A_1 A_2$, where $\eta=\omega^q$, the fact that $A_i^d=\mathbbm{1}$ one finds that
\begin{eqnarray} \label{eq:Ad}
&&\hspace{-0.5cm}(\lambda_{1} A_{1}+\lambda_{2} A_{2})^{d}\nonumber\\
&&=\sum_{k=1}^{d-1}A_{1}^{d-k}A_{2}^{k}\lambda_{1}^{d-k}\lambda_{2}^{k}\sum_{i_{1}=0}^{d-k}\dots\sum_{i_{k}=i_{k-1}}^{d-k}\eta^{\sum_{j=1}^{k}i_{j}}\nonumber\\
&&\hspace{0.5cm}+ A_{1}^{d}\lambda_{1}^{d}+A_{2}^{d}\lambda_{2}^{d}.
\end{eqnarray}
Using Eq. (\ref{warunek}) proven below in Lemma \ref{lem_sum_omega}, the above considerably simplifies to 
\begin{eqnarray} \label{eq:Ad2}
(\lambda_{1} A_{1}+\lambda_{2} A_{2})^{d}=A_{1}^{d}\lambda_{1}^{d}+A_{2}^{d}\lambda_{2}^{d},
\end{eqnarray}
which due to the fact that $A_i^d=\mathbbm{1}$ leads us directly to Eq. (\ref{eq_lemma_ad}), completing the proof.
\end{proof}

\begin{lem}\label{lem_sum_omega}
Given $\eta=\omega^{-q}$ for $q\in \{1,\dots,d-1\}$, if $q$ and $d$ are coprime, the following holds true
\begin{equation}\label{warunek}
\sum_{i_{1}=0}^{d-k}\sum_{i_2=i_1}^{d-k}\dots\sum_{i_{k}=i_{k-1}}^{d-k}\eta^{\sum_{j=1}^{k}i_{j}}=0,
\end{equation}
where $k\in \{1,\dots,d-1\}$.
\end{lem}
\begin{proof}
Let us consider the expression on the left-hand side of Eq. (\ref{warunek}) and shift the summation index in the last sum by $k$. That is, we apply the transformation $i_k\to i_k-k$, which leads us to 
\begin{equation}
    \sum_{i_{1}=0}^{d-k}\dots\sum_{i_{k}=i_{k-1}}^{d-k}\eta^{\sum_{j=1}^{k}i_{j}}=\eta^{-k}\sum_{i_{1}=0}^{d-k}\dots\sum_{i_{k}=i_{k-1}+k}^{d}\eta^{\sum_{j=1}^{k}i_{j}}.
\end{equation}
We then implement a similar transformation to the last but one sum, but this time 
we shift the index by $k-1$ [$i_{k-1}\to i_{k-1}-(k-1)$],
which allows us to rewrite the above as
\begin{eqnarray}
    &&\hspace{-0.5cm}\eta^{-k}\sum_{i_{1}=0}^{d-k}\dots\sum_{i_{k-1}=i_{k-2}}^{d-k}\sum_{i_{k}=i_{k-1}+k}^{d}\eta^{\sum_{j=1}^{k}i_{j}}\nonumber\\
    &&\hspace{-0.5cm}=\eta^{-k-(k-1)}\sum_{i_{1}=0}^{d-k}\dots\sum_{i_{k-1}=i_{k-2}+k-1}^{d-1}\sum_{i_{k}=i_{k-1}+1}^{d}\eta^{\sum_{j=1}^{k}i_{j}}.\nonumber\\
\end{eqnarray}
Notice that this action changes also the summation range of the last sum. 

We then recursively shift the remaining summation indices as $i_j\to i_j-j$ for $j=1,\ldots,k-2$, which yields the following expression
\begin{eqnarray}
  &&\sum_{i_{1}=0}^{d-k}\sum_{i_2=i_1}^{d-k}\dots\sum_{i_{k}=i_{k-1}}^{d-k}\eta^{\sum_{j=1}^{k}i_{j}}\nonumber\\
 && =\eta^{-k(k+1)/2}\sum_{i_{1}=1}^{d-k+1}\;\sum_{i_{2}=i_{1}+1}^{d-k+2}\dots\sum_{i_{k}=i_{k-1}+1}^{d}\eta^{\sum_{j=1}^{k}i_{j}},\nonumber\\
\end{eqnarray}
which can also be stated as
\begin{eqnarray}\label{riesling}
  &&\sum_{i_{1}=0}^{d-k}\sum_{i_2=i_1}^{d-k}\dots\sum_{i_{k}=i_{k-1}}^{d-k}\eta^{\sum_{j=1}^{k}i_{j}}\nonumber\\
 && =\eta^{-k(k+1)/2}\sum_{1\leq i_{1}<i_{2}<\dots<i_{k}\leq d}\eta^{\sum_{j=1}^{k}i_{j}}.
\end{eqnarray}
Our aim now is to prove that the above sum vanishes. To this end, we employ
the so-called Vieta's formulas which relate roots $x_j$ of any complex polynomial of degree $n$
\begin{equation}
P_n(x)=\sum_{i=0}^{n}a_{i}x^{i},
\end{equation}
to the coefficients $a_i$ defining it,
\begin{equation}\label{eq_vieta}
\sum_{1\leq i_{i}<i_{2}<\dots<i_{k}\leq n}\left(\prod_{j=1}^{k} x_{i_{j}}\right)=(-1)^{k}\frac{a_{n-k}}{a_{n}}.
\end{equation}
Let us then consider a particular polynomial of degree $d$,
\begin{equation}\label{eq_p_omega}
P_{d}^{\omega}(x)=x^{d}-1,
\end{equation}
whose roots are obviously $\omega^{i}$. However instead of enumerating these roots by powers of $\omega$, i.e., $x_{i}=\omega^{i}$, we can use the fact that $q$ and $d$ are coprime and enumerate them by powers of $\eta$, that is,
\begin{equation}
x_{i}=\eta^{i}.
\end{equation}
Since all coefficients of this polynomial except $a_d$ and $a_0$ are zero, 
Eq. \eqref{eq_vieta} implies that for any $k=1,\ldots,d-1$,
\begin{equation}
\sum_{1\leq i_{i}<i_{2}<\dots<i_{k}\leq d}\;\prod_{j=1}^{k}\eta^{i_{j}}=\sum_{1\leq i_{i}<i_{2}<\dots<i_{k}\leq d}\eta^{\sum_{j=1}^{k}i_{j}}=0.
\end{equation}
Substitution of the above equation to Eq. \eqref{riesling} ends the proof.
\end{proof}

\section{Supplementary Note 3: Detailed example of an inflation method}\label{app:inflation}
Here we present an extensive explanation of how the inflation method works on the example from the main text. 

Consider the triangle network $\mathcal{O}^{\Delta}$ presented in Figure ~\ref{fig_tri}. The most general form of a state that can be generated in such a network under the LOSR assumption is given by
\begin{equation}\label{app:eq_rho}
\rho^{\mathcal{O}^{\Delta}}= \sum_{\lambda}p_{\lambda} \mathcal{E}_{1}^{(\lambda)}\otimes\mathcal{E}_{2}^{(\lambda)}\otimes\mathcal{E}_{3}^{(\lambda)}(\sigma_{1,2}\otimes\sigma_{1,3}\otimes\sigma_{2,3}),
\end{equation}
where $\sigma_{i,j}$ is a state shared by parties $i$ and $j$ and $\mathcal{E}_{i}^{(\lambda)}$ is a local operation performed by the party $i$.

Next, let us fix $d=3$ and consider the graph state $\ket{G_{\Delta}}$ corresponding to the triangle graph shown in Figure \ref{fig_graph}, i.e., one that satisfies
$g_{i}\ket{G_{\Delta}}=\ket{G_{\Delta}}$ for $i\in\{1,2,3\}$, where the stabilizing operators are given by
\begin{equation}\label{eq:graph}
g_{1}=X_{1}Z_{2}^{2}Z_{3},\qquad g_{2}=Z_{1}^{2}X_{2}, \qquad g_{3}=Z_{1}X_{3}.
\end{equation}
In what follows we prove that this state cannot be generated in the network $\mathcal{O}^{\Delta}$. In other words, we demonstrate that $\ket{G_{\Delta}}$ does not admit the form given in Eq. \eqref{app:eq_rho}. We will achieve this goal by assuming that the above graph state can be generated in $\mathcal{O}^{\Delta}$, and then showing that this assumption leads to the violation of a certain inequality that follows from Lemma \ref{app:lem_ineq}. We thus prove the above fact by contradiction.

Let us consider an inflation of $\mathcal{O}^{\Delta}$, denoted $\mathcal{I}_{0}^{\Delta}$ [cf. Figure \ref{fig_tri}], and two operators defined on it, $g_1=X_{1}Z_{2}^{2}Z_{3}$ and $Z_1^2X_{2'}$. Notice, that while the former acts only on the non-primed parties, the latter acts on party $1$, but also on a primed party $2'$. So, importantly, even though $g_{1}$ and $g_{2}$, as defined in Eq. (\ref{eq:graph}) commute, this is not the case for $g_{1}$ and $Z_1^2X_{2'}$: we actually have $g_1 Z_1^2X_{2'}=\omega Z_1^2X_{2'} g_1$. Therefore, these two operators fulfil the conditions of Lemma \ref{app:lem_ineq} with $q=1$, hence we have
\begin{equation}\label{eq_ineq_sq_ex_app}
\sum_{k=0}^{2}\left(\left\langle g_{1}^{k}\right\rangle_{\mathcal{I}_{0}^{\Delta}}+\left\langle (Z_{1}^{2}X_{2'})^{k}\right\rangle_{\mathcal{I}_{0}^{\Delta}}\right)\leqslant 3 + \sqrt{3},
\end{equation}
where $\langle \cdot\rangle_{\mathcal{I}_{0}^{\Delta}}$ denotes the expected value calculated on a general state that can be generated in $\mathcal{I}_{0}^{\Delta}$. For completeness, let us notice that such as general state has the following form 
\begin{equation}
\rho^{\mathcal{I}_{0}^{\Delta}}= \sum_{\lambda}p_{\lambda} \mathcal{E}_{1}^{(\lambda)}\otimes\mathcal{E}_{2}^{(\lambda)}\otimes\mathcal{E}_{3}^{(\lambda)}\otimes\mathcal{E}_{1'}^{(\lambda)}\otimes\mathcal{E}_{2'}^{(\lambda)}\otimes\mathcal{E}_{3'}^{(\lambda)}(\sigma^{\mathcal{I}_{0}^{\Delta}}),
\end{equation}
where
\begin{equation}
\sigma^{\mathcal{I}_{0}^{\Delta}}=\sigma_{1,2}\otimes\sigma_{1,3}\otimes\sigma_{2,3}\otimes\sigma_{1',2'}\otimes\sigma_{1',3'}\otimes\sigma_{2',3'}.
\end{equation}
Our goal now is to show that the assumption that the graph state $\ket{G}_{\Delta}$ can be generated in $\mathcal{O}^{\Delta}$ leads to a contradiction with 
inequality \eqref{eq_ineq_sq_ex_app}. We achieve it by proving that, under this assumption, every expected value in \eqref{eq_ineq_sq_ex_app} equals one.

Let us focus on the first expected value in (\ref{eq_ineq_sq_ex_app}). Given that $\ket{G_{\Delta}}$ can be prepared in $\mathcal{O}^{\Delta}$, by tracing out all primed parties we get
\begin{equation}\label{eq3}
 \langle g_{1}^{k}\rangle_{\mathcal{I}_{0}^{\Delta}}=\langle g_{1}^{k}\rangle_{\mathcal{O}^{\Delta}}=1.
\end{equation}
This can also be seen by directly calculating the reduced state
of $\rho^{\mathcal{I}_0^{\Delta}}$ corresponding to the non-primed parties, which 
is exactly the state that can be prepared in the original network, that is, 
$\rho^{\mathcal{O}^{\Delta}}$: 
\begin{eqnarray}\label{eq:rho}
\rho^{\mathcal{I}_{0}^{\Delta}}_{1,2,3}&=&\operatorname{Tr}_{1',2',3'}\left(\rho^{\mathcal{I}_{0}^{\Delta}}\right)\nonumber\\
&=& \sum_{\lambda}p_{\lambda} \mathcal{E}_{1}^{(\lambda)}\otimes\mathcal{E}_{2}^{(\lambda)}\otimes\mathcal{E}_{3}^{(\lambda)}(\sigma_{1,2}\otimes\sigma_{1,3}\otimes\sigma_{2,3})\nonumber\\
&=&\rho^{\mathcal{O}^{\Delta}}.
\end{eqnarray}
Having this we can write that
\begin{equation}
\langle g_{1}^{k}\rangle_{\mathcal{I}_{0}^{\Delta}}= \operatorname{Tr}\left(g_{1}^{k}\rho^{\mathcal{I}_{0}^{\Delta}}\right)
=\operatorname{Tr}\left(g_{1}^{k}\rho^{\mathcal{I}_{0}^{\Delta}}_{1,2,3}\right)
=\operatorname{Tr}\left(g_{1}^{k}\rho^{\mathcal{O}^{\Delta}}\right),
\end{equation}
where the second equality follows from the fact that 
$g_1$ acts only on the non-primed parties, whereas the third equality is
a consequence of \eqref{eq:rho}. We can then employ the assumption that 
the state generated by $\mathcal{O}^{\Delta}$ is our graph state, i.e.,
$\rho^{\mathcal{O}^{\Delta}}=\proj{G_{\Delta}}$, and therefore
\begin{equation}
\langle g_{1}^{k}\rangle_{\mathcal{I}_{0}^{\Delta}}=
\operatorname{Tr}\left(g_{1}^{k}\ket{G_{\Delta}}\!\!\bra{G_{\Delta}}\right)=\operatorname{Tr}\left(\ket{G_{\Delta}}\!\!\bra{G_{\Delta}}\right)=1,
\end{equation}
where the second equality stems from the fact that $g_1$ stabilizes $\ket{G_{\Delta}}$, i.e., $g_1\ket{G_{\Delta}}=\ket{G_{\Delta}}$.

As for the second expected value in \eqref{eq_ineq_sq_ex_app}, 
we consider another inflation $\mathcal{I}_{1}^{\Delta}$ of $\mathcal{O}^{\Delta}$ [cf. Figure \ref{fig_tri}] and use Fact \ref{app:fact_stab} (stated below in Supplementary Note \hyperref[app_proof]{4} and in the main text) to show that
\begin{equation}
\left\langle Z_{1}^{2}X_{2'}\right\rangle_{\mathcal{I}_{0}^{\Delta}}=\left\langle g_{2}\right\rangle_{\mathcal{I}_{1}^{\Delta}},
\end{equation}
where $\langle\cdot \rangle_{\mathcal{I}_{1}^{\Delta}}$
stands for an expected value calculated on the state that can be prepared in 
${\mathcal{I}_{1}^{\Delta}}$; notice that the two-body reduced density matrix corresponding to nodes $1$ and $2'$ of the state preparable ${\mathcal{I}_{0}^{\Delta}}$ is the same as the one corresponding to the nodes $1$ 
and $2$ of the state that can be generated in ${\mathcal{I}_{1}^{\Delta}}$ and therefore the above relation follows.

Let us now prove that $\langle
g_{2}\rangle_{\mathcal{I}_{1}^{\Delta}}=1$. Since
the nodes $1$ and $2$ are disconnected in $\mathcal{I}_1^{\Delta}$, but 
are connected in $\mathcal{O}^{\Delta}$ we cannot
directly obtain this expected value from the original network
$\mathcal{O}^{\Delta}$ and the state $\ket{G_{\Delta}}$. However, we can 
compute it indirectly by employing Fact \ref{app:fact_stab}. The idea is to link $\langle g_{2}\rangle_{\mathcal{I}_{1}^{\Delta}}$ to two other expected values. 
If the latter are chosen appropriately then we will be able to calculate them in the same way as we calculated the value of $\langle g_{1}\rangle_{\mathcal{I}_{0}^{\Delta}}$.

Notice first that it follows from Fact \ref{app:fact_stab} that if 
\begin{equation}\label{tintodeverano}
    \langle g_{2}g_{3}\rangle_{\mathcal{I}_{1}^{\Delta}}=\langle g_{3}^{2}\rangle_{\mathcal{I}_{1}^{\Delta}}=1,
\end{equation}
then $\left\langle g_{2}\right\rangle_{\mathcal{I}_{1}^{\Delta}}=1$;
notice that, by the very definition, $g_2$ and $g_3$ commute and
they give the identity when raised to the third power and thus the assumptions
of Fact \ref{app:fact_stab} are satisfied.
The above observation allows us to shift the focus from $\left\langle g_{2}\right\rangle_{\mathcal{I}_{1}^{\Delta}}$ to the expected values appearing in Eq. (\ref{tintodeverano}), that is, $\langle g_{2}g_{3}\rangle_{\mathcal{I}_{1}^{\Delta}}$ and $\langle g_{3}^{2}\rangle_{\mathcal{I}_{1}^{\Delta}}$.

Now, by teaching out all primed parties we get
\begin{equation}\label{eq1}
 \langle g_{2}g_{3}\rangle_{\mathcal{I}_{1}^{\Delta}}=\langle g_{2}g_{3}\rangle_{\mathcal{O}^{\Delta}}=1,   
\end{equation}
as well as that
\begin{equation}\label{eq2}
    \langle g_{3}^{2}\rangle_{\mathcal{I}_{1}^{\Delta}}=\langle g_{3}^{2}\rangle_{\mathcal{O}^{\Delta}}=1.
\end{equation}
(Notice that the second equations (\ref{eq1}) and (\ref{eq2}) stem directly from
the fact that both $g_2$ and $g_3$ stabilize the graph state $\ket{G_{\Delta}}$.)
Since now the relations in Eq. (\ref{tintodeverano}) hold true, we conclude, 
as already mentioned, $\left\langle g_{2}\right\rangle_{\mathcal{I}_{1}^{\Delta}}=1$, which means that $\left\langle Z_{1}^{2}X_{2'}\right\rangle_{\mathcal{I}_{0}^{\Delta}}=1$.

By the same argument we also have that $\left\langle Z_{1}X_{2'}^{2}\right\rangle_{\mathcal{I}_{0}^{\Delta}}=1$. This together with  (\ref{eq3}) imply that the left-hand side of \eqref{eq_ineq_sq_ex_app} equals six, leading to a contradiction. Thus, the triangle graph state cannot be prepared in the network $\mathcal{O}^{\Delta}$.

\section{Supplementary Note 4: Proof of Theorem 1}
\label{app_proof}

Here we provide a full proof of Theorem 1 stated in the main text 
which for completeness we recall here. 
\begin{thm}\label{app:thm_bip}
Consider a graph $G$ with $N\geqslant 3$ vertices and where at least one vertex $i$ has a neighbourhood  $|\mathcal{N}_{i}|\geqslant 2$. A graph state $\ket{G}\in \mathbb{C}_{d}^{\otimes N}$, where $d$ is prime, corresponding to a graph $G$ cannot be generated in an LOSR $N$-partite quantum network with bipartite sources.
\end{thm}
Before we proceed with the proof
let us recall a couple of important notions.

First, if a graph $G$ can be transformed into a graph $G'$ using a set of transformations called local complementations, then there exists a set of local unitaries $U_{i}$ such that
\begin{equation}
\ket{G}=\bigotimes_{i=1}^{N}U_{i} \ket{G'}
\end{equation}
(see Theorem 5 in \cite{bahramgiri2007graph}). The local complementations on the vertex $n$ consist of the transformations of a graph $G$ to a graph $G'$ such that
\begin{equation}
\Gamma_{i,j}'= \Gamma_{i,j} + a_{n}\Gamma_{i,n}\Gamma_{j,n}
\end{equation}
for $i\neq j$, where $a_{n}\in\{0,\dots, d-1\}$, and addition and multiplications are performed modulo $d$.

Next, let $\mathcal{O}$ be a network and $B=\bigotimes_{i\in \mathcal{O}}B_{i}$. We call a subnetwork of $\mathcal{O}$ the $B$ subnetwork if it consists only of parties $i$ for which $B_{i}\neq \mathbb{1}$ and sources that distribute states to said parties. For example, if $\mathcal{O}$ is a tripartite, fully-connected network and $B=X\otimes X\otimes \mathbb{1}$, then the $B$ subnetwork of $\mathcal{O}$ is a network consisting of parties $1$ and $2$ and a source distributing a bipartite state to those parties. This definition does not introduce anything new as one can notice that in Fact \ref{app:fact_=_1} we have already used these types of subnetworks. What this definition does, is that facilitates the application of Fact \ref{app:fact_=_1}  which will be of great help in the proof of Theorem \ref{app:thm_bip}.

Let us now also recall two facts stated in the main text which will be
frequently used in the proof.
First, let us introduce an isomorphism $\chi:\mathcal{I}\rightarrow\mathcal{I}'$ with an associated set $S_{\chi}\subset \{1,\dots, N\}$, that acts by swapping labels of parties $i$ and $i'$ for all $i\in S_{\chi}$. If an operator $M$ is a $2N$-fold tensor product, then we use the notation $\chi(M)$ for a swap operation between parties $i$ and $i'$ for all $i\in S_{\chi}$.
\begin{fact}\label{app:fact_=_1}
Consider a network $\mathcal{O}$ and two different inflations of it, $\mathcal{I}_{1}$ and $\mathcal{I}_{2}$. Consider also two matrices $B=\bigotimes_{i\in \mathcal{I}_{1}}B_i$ and $C=\bigotimes_{i\in \mathcal{I}_{2}}C_{i}$ that act nontrivially on some subnetworks $\mathcal{I}_i'\subseteq\mathcal{I}_i$. Then, $\langle B\rangle_{\mathcal{I}_{1}}=\langle C\rangle_{\mathcal{I}_{2}}$ if there exists an isomorphism $\chi$ such that
\begin{equation}
\chi(\mathcal{I}'_1)=\mathcal{I}'_2 \quad \textrm{and}\quad \chi(B)=C.
\end{equation}
\end{fact}

\begin{fact}\label{app:fact_stab}
Consider three mutually commuting unitary matrices $B_i$ that obey $B_{i}^{d}=\mathbb{1}$. If $\langle B_{1}B_{3}\rangle = \langle B_{2}B_{3}^{\dagger}\rangle = 1$, then $\langle B_{1}B_{2}\rangle=1$.
\end{fact}

Finally, it will be convenient to consider certain sets of graph states separately. This division into sets is done with the following lemma.

\begin{lem}\label{app:lem_graph}
Every graph that has at least $3$ vertices and has at least one vertex with a neighbourhood $|\mathcal{N}_{i}|\geqslant 2$ can be transformed using local complementations and relabelling into a graph $G$ that fulfils $|\mathcal{N}_{1}\setminus\mathcal{N}_{2}|\geqslant 2$, $\Gamma_{1,2}\neq 0$ and one of the four following sets of conditions:
\begin{enumerate}
    \item $\mathcal{N}_{1}\cap \mathcal{N}_{2}=\emptyset$,
    \item  $\mathcal{N}_{1}\cap \mathcal{N}_{2}\neq\emptyset$ and for all $n\in \mathcal{N}_{1}\cap\mathcal{N}_{2}$ we have $\mathcal{N}_{2}\setminus\{n\}\neq \mathcal{N}_{n}\setminus\{2\}$,
    \item $|\mathcal{N}_{1}\cap \mathcal{N}_{2}|=1$ and there exists $n\in \mathcal{N}_{1}\cap\mathcal{N}_{2}$ such that $\mathcal{N}_{2}\setminus\{n\}= \mathcal{N}_{n}\setminus\{2\}$,
    \item $|\mathcal{N}_{1}\cap \mathcal{N}_{2}|\geqslant 2$, there exists $n\in \mathcal{N}_{1}\cap\mathcal{N}_{2}$ such that $\mathcal{N}_{2}\setminus\{n\}= \mathcal{N}_{n}\setminus\{2\}$, and for each such $n$ there exists $a_{n}\in\{1,\dots,d-1\}$ such that for all $i\in\mathcal{N}_{n}\setminus\{2\}$ we have $\Gamma_{2,i}+a_{n}\Gamma_{2,n}\Gamma_{n,i}=0$.
\end{enumerate}
\end{lem}
\begin{proof}
We assume that a graph $G$ has at least $3$ vertices and that at least one vertex has a neighbourhood $|\mathcal{N}_{i}|\geqslant 2$, and so, without a loss of generality, we can take $|\mathcal{N}_{1}|\geqslant 2$ and $\Gamma_{1,2}\neq 0$. Then we act $a_{1}$ times with a local complementation on vertex $1$, which gives us
\begin{equation}
\Gamma_{2,n}'=\Gamma_{2,n}+a_{1}\Gamma_{1,2}\Gamma_{1,n},
\end{equation}
where $n\in \mathcal{N}_{1}\setminus \{2\}$. We choose $a_{1}$ such that $\Gamma_{2,n}'=0$, and since local complementation on $1$ does not change $\Gamma_{1,2}$, it implies $|\mathcal{N}_{1}\setminus \mathcal{N}_{2}|\geqslant 2$. If after this operation we have $\mathcal{N}_{1}\cap \mathcal{N}_{2}=\emptyset$ then the first set of conditions in the lemma is fulfilled.

As for the case when $\mathcal{N}_{1}\cap \mathcal{N}_{2}\neq\emptyset$,  the graph $G$ fulfils the second set of conditions in the lemma if for all $n\in \mathcal{N}_{1}\cap \mathcal{N}_{2}$  we have
\begin{equation}\label{eq_2n}
\mathcal{N}_{2}\setminus\{n\}\neq \mathcal{N}_{n}\setminus\{2\}.
\end{equation}

Let us now consider a case where for some $n\in \mathcal{N}_{1}\cap \mathcal{N}_{2}$ the above equation does not hold true. If $|\mathcal{N}_{1}\cap \mathcal{N}_{2}|=1$ then the third set of conditions in the lemma is fulfilled. Otherwise, let us consider a local complementation on the vertex $n$
\begin{equation}
\Gamma_{2,i}'=\Gamma_{2,i}+a_{n}\Gamma_{2,n}\Gamma_{n,i},
\end{equation}
where $i\in \mathcal{N}_{n}\setminus \{2\}$. We assume that \eqref{eq_2n} does not hold for $n$, hence for every $i$ we also have $i\in \mathcal{N}_{2}$. If one can find $a_{n}$ such that $\Gamma_{2,i}'=0$ for all $i\in \mathcal{N}_{n}\setminus \{2\}$, then the fourth set of assumptions from the lemma is fulfilled.

Conversely, if choosing such an $a_{n}$ is not possible, then we can take $a_{n}$ such that $\Gamma_{1,2}'\neq 0$ and such that there exists $i\in\mathcal{N}_{2}\setminus \{n\}$ for which $\Gamma_{2,i}'=0$. Clearly, after this transformation we get 
\begin{equation}\label{eq_n2n2}
|\mathcal{N}_{2}'| < |\mathcal{N}_{2}|.
\end{equation}
Let us investigate the properties of this transformed graph. First of all $\Gamma_{1,2}\neq 0$, since that was the condition on $a_{n}$. Next, the local complementation on the vertex $n$ by definition does not change $\Gamma_{i,n}$ for all $i$, so in particular we still have $\Gamma_{1,n}\neq 0$ and $\Gamma_{2,n}\neq 0$ which implies $\mathcal{N}_{1}\cap \mathcal{N}_{2}\neq \emptyset$. Last but not least, the action of a local complementation on vertex $n\in \mathcal{N}_{1}\cap \mathcal{N}_{2}$ cannot decrease the size of a set $\mathcal{N}_{1}\setminus \mathcal{N}_{2}$ meaning, that the assumption $|\mathcal{N}_{1}\setminus \mathcal{N}_{2}|\geqslant 2$ still holds. 

To see why the last claim holds true, notice that to decrease the size of a set $\mathcal{N}_{1}\setminus \mathcal{N}_{2}$ we would have to have $\Gamma_{1,j}\neq 0$ and $\Gamma_{2,j}=0$ for some $j$ before the action of the local complementation on $n$, and after said action either $\Gamma'_{1,j}= 0$ or $\Gamma'_{2,j}\neq0$. However, from the assumption $\mathcal{N}_{2}\setminus\{n\}= \mathcal{N}_{n}\setminus\{2\}$ it follows that $\Gamma_{n,j}=0$. Crucially, this implies that under the action of local complementation on $n$ we have
\begin{equation}
\Gamma'_{j,k}=\Gamma_{j,k}
\end{equation}
for all $k$, hence $\Gamma'_{1,j}= \Gamma_{1,j}$ and $\Gamma'_{2,j}=\Gamma_{2,j}$.

The key observation here is that the new graph created after action of local complementation fulfils $|\mathcal{N}_{1}\setminus \mathcal{N}_{2}|\geqslant 2$, $\Gamma_{1,2}\neq 0 $ and $\mathcal{N}_{1}\cap \mathcal{N}_{2}\neq \emptyset$, and so we can check if the new graph fulfils conditions 2, 3 or 4 from the lemma and if they are not fulfilled, then we again create a new graph by action of local complementation on some $\tilde{n}\in \mathcal{N}_{1} \cap \mathcal{N}_{2}$. Every time we act with this local complementation we get \eqref{eq_n2n2}, which implies that after a finite amount of repetition of this procedure, we will get $|\mathcal{N}_{1}\cap \mathcal{N}_{2}|=1$ which fulfils either second or third condition from the lemma.
\end{proof}

With that, we are ready to prove Theorem \ref{app:thm_bip}.
\begin{proof}
We assume that the graph $G$ consists of at least three vertices and has at least one vertex $i$ such that $|\mathcal{N}_{i}|\geqslant 2$. So by the virtue of Lemma \ref{app:lem_graph} by using local complementations and relabelling we can transform $G$ into a graph that fulfils one of the four sets of conditions. Acting with a local complementation on a graph $G$ corresponds to acting with local unitaries on a graph state $\ket{G}$, and so it is sufficient to only consider graphs that fulfil one of the four sets of conditions in Lemma \ref{app:lem_graph}. 

Since every inflation in this proof has a very similar structure, let us begin by discussing these similarities in order to make the proof easier to follow. All inflations $\mathcal{I}_{k}^{l}$ that we consider in this proof share two assumptions. First, every $\mathcal{I}_{k}^{l}$ consists of exactly two copies of every party, denoted $i$ and $i'$, and two copies of every source from $\mathcal{O}$. Second, a non-primed party $i$ (for $i\neq 2$) is connected to every other non-primed party $j$ (for $j\neq 2$) and to either $2$ or $2'$. Likewise, every primed party $i'$ (for $i'\neq2'$) has to be connected to every other primed party $j'$ (for $j'\neq 2'$) and to either $2$ or $2'$. These two assumptions allow us to precisely describe an inflation just by defining the set of non-primed parties connected to $2'$, which we denote by $T_{k}^{l}$. For a graphical description, see Figure ~\ref{fig_I_k}.

We structure this proof in the following way: in \textit{Part 1} we consider graphs $G$ fulfilling the first or the second set of conditions from Lemma \ref{app:lem_graph}, in \textit{Part 2} we consider graphs $G$ fulfilling the third set of conditions and in \textit{Part 3} we consider graphs fulfilling the fourth set of conditions.

\textit{Part 1.} 
Here, we assume that a graph $G$ fulfils the first or second set of conditions from Lemma \ref{app:lem_graph}. Both these sets of conditions imply the same structure of the first two stabilizing operators, namely:
\begin{align}\label{eq_g1_g2}
\begin{split}
g_{1}&=X_{1}Z_{2}^{\Gamma_{1,2}}Z_{\mathcal{N}_{1}\setminus \{2\}},\\
g_{2}&=Z_{1}^{\Gamma_{1,2}}X_{2}Z_{\mathcal{N}_{2}\setminus \{1\}},
\end{split}
\end{align}
where 
\begin{equation}
Z_{\mathcal{N}_{i}\setminus\{j\}}=\prod_{k\in \mathcal{N}_{i}\setminus\{j\}}Z_{k}^{\Gamma_{i,k}}.
\end{equation}

Let us consider an inflation $\mathcal{I}^{1}_{0}$ for which
\begin{equation}\label{eq_t0}
T^{1}_{0}=\Big(\mathcal{N}_{1}\setminus(\mathcal{N}_{2}\cup \{2\})\Big) \cup \{1\}.
\end{equation}
By the virtue of Lemma \ref{app:lem_ineq}, the following holds true:
\begin{equation}\label{eq_square}
\sum_{k=0}^{d-1}\left\langle \tilde{g}_{1}^{k}\right\rangle_{\mathcal{I}_{0}^{1}} + \left\langle g_{2}^{k} \right\rangle_{\mathcal{I}^{1}_{0}} \leq d+\sqrt{d},
\end{equation}
where $\tilde{g}_{1}=X_{1}Z_{2'}^{\Gamma_{1,2}}Z_{\mathcal{N}_{1}\setminus \{2\}}$. To complete the proof we will show that the above inequality is violated, since if one could generate the aforementioned graph state in the network $\mathcal{O}$, then the above operators stabilize the state generated in the inflation $\mathcal{I}_{0}^{1}$. We begin our analysis with the operator $\tilde{g}_{1}$. To show that this operator is a stabilizing operator, let us consider another inflation $\mathcal{I}^{1}_{1}$ for which
\begin{equation}
T^{1}_{1}=\mathcal{N}_{1}\cap \mathcal{N}_{2}.
\end{equation}
Notice that the union
\begin{equation}
T^{1}_{0}\cup T^{1}_{1}\cup\{2'\}= (\mathcal{N}_{1}\setminus\{2\})\cup \{1,2'\}
\end{equation}
describes a set of parties in a $\tilde{g}_{1}$ subnetwork of $\mathcal{I}^{1}_{0}$. In this subnetwork, every party from the set $T_{0}^{1}$ is connected to $2'$ and to $T^{1}_{1}$, but $2'$ and $T^{1}_{1}$ are unconnected. Similarly, in the inflation $\mathcal{I}^{1}_{1}$, every party from the set $T^{1}_{0}$ is connected to $2$ and to $T^{1}_{1}$, but $2$ and $T_{1}^{1}$ are unconnected, hence by the virtue of Fact \ref{app:fact_=_1} we have
\begin{equation}\label{eq_g1_i0i1}
\left\langle   \tilde{g}_{1}\right\rangle_{\mathcal{I}^{1}_{0}}=\langle g_{1} \rangle_{\mathcal{I}^{1}_{1}}.
\end{equation}

The first set of assumptions from Lemma \ref{app:lem_graph} implies $T^{1}_{1}=\emptyset$ which gives us $\langle g_{1} \rangle_{\mathcal{I}^{1}_{1}}=\langle g_{1} \rangle_{\mathcal{O}}=1$, since we assume that we generate a graph state $\ket{G}$ in the original network $\mathcal{O}$. Of course, for the second set of assumptions we have $T^{1}_{1}\neq\emptyset$, and so the same implication does not hold. Therefore, for the next few paragraphs we only consider the second set of assumptions from Lemma \ref{app:lem_graph} and we will return to considering both after showing that $\langle g_{1} \rangle_{\mathcal{I}^{1}_{1}}=\langle g_{1} \rangle_{\mathcal{O}}=1$ holds also for the second set of assumptions.

One of the assumptions from the second set of assumptions from Lemma \ref{app:lem_graph} states that for all $n\in \mathcal{N}_{1}\cap\mathcal{N}_{2}$ we have
\begin{equation}\label{eq_N2_Nn}
\mathcal{N}_{2}\setminus \{n\} \neq \mathcal{N}_{n}\setminus\{2\}.
\end{equation}
Consequently, given two stabilizing operators $g_{2}$ and $g_{n}$ there has to exist a qubit $i$ such that $i\neq 2,n$ and $g_{2}^{(i)}\neq g_{n}^{(i)}$, where $g_{k}^{(i)}$ denotes a matrix of $g_{k}$ acting on the subsystem corresponding to the party $i$. Utilising the symmetry: $g^{(j)}_{i}=g^{(i)}_{j}$, we get two possibilities: either
\begin{equation}\label{eq_gi_1}
g_{i}^{(2)}=\mathbb{1},\qquad g_{i}^{(n)}=Z^{\Gamma_{i,n}}
\end{equation}
or
\begin{equation}\label{eq_gi_2}
g_{i}^{(2)}=Z^{\Gamma_{2,i}},\qquad g_{i}^{(n)}=\mathbb{1},
\end{equation}
where $\Gamma_{i,n},\Gamma_{2.i}\neq0$. In what follows, we assume that Eq. \eqref{eq_gi_1} is true, and we will later come back to the case of Eq. \eqref{eq_gi_2}. 

By the virtue of Fact \ref{app:fact_stab} we have an implication
\begin{equation}
\langle g_{i}^{-l}\rangle_{\mathcal{I}^{1}_{1}}=\langle g_{1}g_{i}^{l} \rangle_{\mathcal{I}^{1}_{1}}=1 \quad \Rightarrow \quad \langle g_{1}\rangle_{\mathcal{I}^{1}_{1}}=1,
\end{equation}
where $l\in\{1,\dots,d-1\}$. 

From \eqref{eq_gi_1} it follows that the $g_{i}^{-l}$ subnetwork of $\mathcal{I}^{1}_{1}$ does not contain party $2$. Remember that we choose inflations such that every non-primed party is connected to every other party with an exception of party $2$, meaning that if a subnetwork does not contain the party $2$ or $2'$ then it is fully connected, which is exactly the case for the $g_{i}^{-l}$ subnetwork of $\mathcal{I}^{1}_{1}$. Consequently, by tracing out all parties on which $g_{i}^{-l}$ acts trivially we get
\begin{equation}\label{eq_gi_I1_rho}
\langle g_{i}^{-l}\rangle_{\mathcal{I}^{1}_{1}} = \langle g_{i}^{-l}\rangle_{\mathcal{O}}=1,
\end{equation}
where in the second equality we used the assumption that we generate a graph state on the network $\mathcal{O}$. As for the expected value $\langle g_{1}g_{i}^{l}\rangle_{\mathcal{I}^{1}_{1}}$, in order to calculate it, we first have to introduce another inflation $\mathcal{I}^{1}_{2}$ for which 
\begin{equation}\label{eq_t2}
T^{1}_{2}=(\mathcal{N}_{1}\cap\mathcal{N}_{2})\setminus\{n\}.
\end{equation}
Notice, that the only difference between inflations $\mathcal{I}^{1}_{1}$ and $\mathcal{I}^{1}_{2}$ is that in the former, parties $2'$ and $n$ are connected and in the latter $2$ and $n$. Crucially, since $d$ is prime and $\Gamma_{i,n}\neq 0$, we can always take
\begin{equation}
l=  -\Gamma_{1,n}\Gamma_{i,n}^{-1},
\end{equation}
where we take the inverse and the negation to be operations in $\mathbb{Z}_{d}$, i.e., $l\in\{1,\dots,d-1\}$. This implies that the $g_{1}g_{i}^{l}$ subnetwork of $\mathcal{I}^{1}_{1}$ does not contain party $n$, therefore by Fact \ref{app:fact_=_1} we have
\begin{equation} \label{eq_inflation_step}
\langle g_{1}g_{i}^{l}\rangle_{\mathcal{I}^{1}_{1}}=\langle g_{1}g_{i}^{l}\rangle_{\mathcal{I}^{1}_{2}}.
\end{equation}
Lastly, by using Fact \ref{app:fact_stab} we can show that if $\langle g_{1}\rangle_{\mathcal{I}^{1}_{2}}=\langle g_{i}^{l}\rangle_{\mathcal{I}^{1}_{2}}=1$ then $\langle g_{1}g_{i}^{l}\rangle_{\mathcal{I}^{1}_{2}}=1$. Furthermore, since $g_{i}^{l}$ subnetwork of $\mathcal{I}^{1}_{2}$ does not contain party $2$, we have $\langle g_{i}^{l}\rangle_{\mathcal{I}^{1}_{2}}=\langle g_{i}^{l}\rangle_{\mathcal{O}}=1$, and so:
\begin{equation}
\langle g_{1}\rangle_{\mathcal{I}^{1}_{2}}=1 \;\Rightarrow\;  \langle g_{1}g_{i}^{l}\rangle_{\mathcal{I}^{1}_{2}}=1 \;\Rightarrow\; \langle g_{1}\rangle_{\mathcal{I}^{1}_{1}}=1.
\end{equation}

As for the case when Eq. \eqref{eq_gi_2} is true, the proof has mostly the same structure, but with $g_{i}^{-l}$ and $g_{1}g_{i}^{l}$ swapped. The only other difference is that, using the fact that $d$ is prime and $\Gamma_{2,i}\neq 0$, we take $l\in \mathbb{Z}_{d}$ such that
\begin{equation}
 l= -\Gamma_{1,2}\Gamma_{2,i}^{-1}.
\end{equation}
As a result, we have the same implication as in the first case: 
\begin{equation}
\langle g_{1}\rangle_{\mathcal{I}^{1}_{2}}=1 \;\Rightarrow\;  \langle g_{1}\rangle_{\mathcal{I}^{1}_{1}}=1.
\end{equation}

Since both cases result in the above implication, we can now take a different $n$ and apply the same procedure to produce an implication
\begin{equation}
\langle g_{1}\rangle_{\mathcal{I}^{1}_{3}}=1 \;\Rightarrow\;  \langle g_{1}\rangle_{\mathcal{I}^{1}_{2}}=1,
\end{equation}
where $T^{1}_{3}=(\mathcal{N}_{1}\cap\mathcal{N}_{2})\setminus\{n,n'\}$. Repeating this procedure for all $n\in T_{1}$ results in a chain of implications
\begin{equation}\label{eq_imp_chain}
\langle g_{1}\rangle_{\mathcal{I}^{1}_{q}}=1 \;\Rightarrow\;  \langle g_{1}\rangle_{\mathcal{I}^{1}_{q-1}}=1 \;\Rightarrow\dots\Rightarrow\; \langle g_{1}\rangle_{\mathcal{I}^{1}_{1}}=1,
\end{equation}
where $q=|\mathcal{N}_{1}\cap \mathcal{N}_{2}|+1$. However, for $\mathcal{I}^{1}_{q}$ we have $T^{1}_{q}=\emptyset$, and so
\begin{equation}
\langle g_{1}\rangle_{\mathcal{I}^{1}_{q}}=\langle g_{1}\rangle_{\mathcal{O}}=1.
\end{equation}
This is the first statement in the chain of implications \eqref{eq_imp_chain}, and so
\begin{equation}\label{eq_g1_I0}
\langle g_{1}\rangle_{\mathcal{I}^{1}_{1}}=1,
\end{equation}
which, by the virtue of Eq. \eqref{eq_g1_i0i1} implies
\begin{equation}\label{eq_x1z2=1}
\langle \tilde{g}_{1}\rangle_{\mathcal{I}^{1}_{0}}=1.
\end{equation}
As a direct implication of the above, we have that $ \tilde{g}_{1}$ is a stabilizing operator of $\mathcal{I}^{1}_{0}$.

We now come back to the inequality \eqref{eq_square} and show that $g_{2}$ is also a stabilizing operator of a state generated in the inflation $\mathcal{I}^{1}_{0}$. We again consider both the first and second set of assumption from Lemma \ref{app:lem_graph}. The $g_{2}$ subnetwork of $\mathcal{I}^{1}_{0}$ is a fully connected subnetwork, apart from the lacking connection between parties $1$ and $2$. Therefore, from Fact \ref{app:fact_=_1} we have that
\begin{equation}\label{eq_g2_i0i-1}
\langle g_{2} \rangle_{\mathcal{I}^{1}_{0}}=\langle g_{2}\rangle_{\mathcal{I}^{1}_{-1}},
\end{equation}
where $\mathcal{I}^{1}_{-1}$ is an inflation with
\begin{equation}
T^{1}_{-1}=\{1\}.
\end{equation}
From $|\mathcal{N}_{1}\setminus\mathcal{N}_{2}|\geq 2$ it follows that there exists $i\in \mathcal{N}_{1}$ for which 
\begin{equation}\label{eq_gi_3}
g_{i}^{(1)}=Z^{\Gamma_{1,i}}, \qquad g_{i}^{(2)}=\mathbb{1},
\end{equation}
where we used the symmetry $g_{i}^{(j)}=g_{j}^{(i)}$. Next, by the virtue of Fact \ref{app:fact_stab} we have
\begin{equation}\label{eq_imp_-1}
\langle g_{i}^{-l}\rangle_{\mathcal{I}^{1}_{-1}}=\langle g_{2}g_{i}^{l}\rangle_{\mathcal{I}^{1}_{-1}}=1 \;\Rightarrow\; \langle g_{2}\rangle_{\mathcal{I}^{1}_{-1}}=1, 
\end{equation}
where we take $l\in \mathbb{Z}_{d}$ such that
\begin{equation}
l= -\Gamma_{1,2}\Gamma_{1,i}^{-1}. 
\end{equation}
Notice, that $\langle g_{2}g_{i}^{l}\rangle_{\mathcal{I}^{1}_{-1}}$ has no nontrivial matrix on the first party and $\langle g_{i}^{-l}\rangle_{\mathcal{I}^{1}_{-1}}$ has no nontrivial matrix on the second party meaning that the $g_{2}g_{i}^{l}$ and $g_{i}^{-l}$ subnetworks of $\mathcal{I}^{1}_{-1}$ are fully connected which by tracing out all parties on which operators act trivially allows us two conclude that
\begin{equation}
\langle g_{i}^{-l}\rangle_{\mathcal{I}^{1}_{-1}}=\langle g_{i}^{-l}\rangle_{\mathcal{O}}=1, \qquad \langle g_{2}g_{i}^{l}\rangle_{\mathcal{I}^{1}_{-1}}=\langle g_{2}g_{i}^{l}\rangle_{\mathcal{O}}=1,
\end{equation}
and so by the virtue of Eq. \eqref{eq_g2_i0i-1} and Eq. \eqref{eq_imp_-1} we have
\begin{equation}
\langle g_{2}\rangle_{\mathcal{I}^{1}_{0}}=1,
\end{equation}
i.e., $g_{2}$ is a stabilizing operator of a state generated in $\mathcal{I}^{1}_{0}$.

As a consequence of the above result and \eqref{eq_g1_I0}, the inequality \eqref{eq_square} is violated, meaning that our assumption was wrong - graph states corresponding to a graph that admit the first or the second set of assumptions from Lemma \ref{app:lem_graph} cannot be generated in a network with bipartite sources.

\textit{Part 2.}
We now assume that the graph $G$ fulfils the third set of conditions from Lemma \ref{app:lem_graph}. In this part we will only use three stabilizing operators, namely
\begin{align}\label{eq_g1g2gn}
\begin{split}
g_{1}&=X_{1}Z_{2}^{\Gamma_{1,2}}Z_{n}^{\Gamma_{1,n}}Z_{\mathcal{N}_{1}\setminus \{2,n\}},\\
g_{2}&=Z_{1}^{\Gamma_{1,2}}X_{2}Z_{n}^{\Gamma_{2,n}}Z_{\mathcal{N}_{2}\setminus \{1,n\}},\\
g_{n}&=Z_{1}^{\Gamma_{1,n}}Z_{2}^{\Gamma_{2,n}}X_{n}Z_{\mathcal{N}_{n}\setminus \{1,2\}},
\end{split}
\end{align}
where $\mathcal{N}_{1}\cap \mathcal{N}_{2} = \{n\}$. Let us consider an inflation $\mathcal{I}_{0}^{2}$ for which $T_{0}^{2}=\mathcal{N}_{2}\setminus\{1,n\}$. We can use Lemma \ref{app:lem_ineq} to construct an inequality:
\begin{equation}
\begin{split}\label{eq_ineq_2}
\sum_{k=0}^{d-1}\expval{ s^{k}}_{\mathcal{I}_{0}^{2}}+ \expval{g_{1}^{k}}_{\mathcal{I}_{0}^{2}}\leqslant d+\sqrt{d},
\end{split}
\end{equation}
where
\begin{equation}
s=X_{2'}Z_{2'}^{l \Gamma_{2,n}}Z_{n}^{\Gamma_{2,n}}X_{n}^{l}Z_{\mathcal{N}_{2}\setminus\{1,n\}}Z_{\mathcal{N}_{n}\setminus\{1,2\}}^{l},
\end{equation}
and we take $l\in \mathbb{Z}_{d}$ such that
\begin{equation}
l=-\Gamma_{1,2}\Gamma_{1,n}^{-1}.
\end{equation}
From the assumption $|\mathcal{N}_{1}\cap\mathcal{N}_{2}|=1$ we can infer that the $g_{1}^{k}$ subnetwork of $\mathcal{I}_{0}^{2}$ does not contain any party from the set $T_{0}^{2}$, hence by tracing out all parties on which $g_{1}^{k}$ acts trivially and by our assumption about the generation of a graph state in the network $\mathcal{O}$ we get
\begin{equation}\label{eq_g1_part2}
\langle g_{1}^{k}\rangle_{\mathcal{I}_{0}^{2}}=\langle g_{1}^{k}\rangle_{\mathcal{O}}=1,
\end{equation}
therefore, for all $k$, $g_{1}^{k}$ is a stabilizing operator of a state generated in $\mathcal{I}^{2}_{0}$.

In order to tackle the second operator from Eq. \eqref{eq_ineq_2}, let us first introduce another inflation $\mathcal{I}_{1}^{2}$ with $T_{1}^{2}=\{n\}$. From Fact \ref{app:fact_=_1} it follows that
\begin{align}
\begin{split}
\left\langle s\right\rangle_{\mathcal{I}_{0}^{2}}=\left\langle g_{2}g_{n}^{l}\right\rangle_{\mathcal{I}_{1}^{2}}.
\end{split}
\end{align}
The above can be easily shown by comparing the two operators and its corresponding subnetworks of $\mathcal{I}_{0}^{2}$ and $\mathcal{I}_{1}^{2}$: the former is fully-connected apart from the lacking connection between $2'$ and $n$, and the latter is fully-connected apart from the lacking connection between $2$ and $n$. 

By the virtue of Fact \ref{app:fact_stab} we have
\begin{equation}
\langle g_{1}^{-m}g_{2}\rangle_{\mathcal{I}_{1}^{2}}=\langle g_{1}^{m}g_{n}^{l}\rangle_{\mathcal{I}_{1}^{2}}=1 \quad \Rightarrow \quad \langle g_{2}g_{n}^{l}\rangle_{\mathcal{I}_{1}^{2}}=1,
\end{equation}
where we take $m\in \mathbb{Z}_{d}$ to be
\begin{equation}
m = \Gamma_{2,n}\Gamma_{1,n}^{-1}.
\end{equation}
A quick look at the $g_{1}^{-m}g_{2}$ and $g_{1}^{m}g_{n}^{l}$ subnetworks of $\mathcal{I}_{1}^{2}$ reveals that both of them are fully-connected, since they do not contain party $n$ and party $2$ respectively, hence by tracing out all parties on which the respective operators act trivially we get
\begin{equation}
\langle g_{1}^{-m}g_{2}\rangle_{\mathcal{I}_{1}^{2}}=\langle g_{1}^{-m}g_{2}\rangle_{\mathcal{O}}=1,
\langle g_{1}^{m}g_{n}^{l}\rangle_{\mathcal{I}_{1}^{2}}=\langle g_{1}^{m}g_{n}^{l}\rangle_{\mathcal{O}}=1, 
\end{equation}
where we used the assumption that we can generate a graph state $\mathcal{G}$ in the network $\mathcal{O}$. However, the above implies
\begin{equation}
\left\langle s \right\rangle_{\mathcal{I}_{0}^{2}}=1,
\end{equation}
and so, as was the case with $g_{1}$, the above operator is a stabilizing operator of a state generated in $\mathcal{I}_{0}^{2}$. Therefore, the inequality \eqref{eq_ineq_2} is violated, which implies that the graph state corresponding to a graph that fulfils the third set of conditions from Lemma \ref{app:lem_graph} cannot be generated in a considered quantum network.

\textit{Part 3.}
In this part, we consider graphs $G$ that fulfil the fourth set of conditions from Lemma \ref{app:lem_graph}. The stabilizing operators used in this part of the proof have the same general structure as the once we have used in the previous part (see Eq.~\eqref{eq_g1g2gn}). 

Let us consider an inflation $\mathcal{I}_{0}^{3}$ for which $T_{0}^{3}=\emptyset$. By the virtue of Lemma \ref{app:lem_ineq} we can construct the following inequality:
\begin{equation}\label{eq_ineq_part3}
\sum_{k=0}^{d-1}\left\langle \left(X_{2'}Z_{2'}^{l\Gamma_{2,n}}Z_{n}^{\Gamma_{2,n}}X_{n}^{l}\right)^{k}+ g_{1}^{k}\right\rangle_{\mathcal{I}_{0}^{3}}\leqslant d+\sqrt{d},
\end{equation}
where $l\in \mathbb{Z}_{d}$. From the assumption that we can generate a graph state $\ket{G}$ in the network $\mathcal{O}$ and by tracing out all parties on which $g_{1}$ acts trivially we get
\begin{equation}\label{eq_g1_part3}
\left\langle g_{1}\right\rangle_{\mathcal{I}_{0}^{3}}=\left\langle g_{1}\right\rangle_{\mathcal{O}}=1,
\end{equation}
and so $g_{1}$ is a stabilizing operator of a state generated in $\mathcal{I}_{0}^{3}$. Next, let us consider an inflation $\mathcal{I}_{1}^{3}$ for which $T_{1}^{3}=\{n\}$. It trivially follows from Fact \ref{app:fact_=_1} that
\begin{equation}
\left\langle X_{2'}Z_{2'}^{l\Gamma_{2,n}}Z_{n}^{\Gamma_{2,n}}X_{n}^{l}\right\rangle_{\mathcal{I}_{0}^{3}}=\left\langle X_{2}Z_{2}^{l\Gamma_{2,n}}Z_{n}^{\Gamma_{2,n}}X_{n}^{l}\right\rangle_{\mathcal{I}_{1}^{3}}.
\end{equation}
Furthermore, one of the assumptions states that there exists $a_{n}\in\{1,\dots,d-1\}$ such that for all $i\in\mathcal{N}_{n}\setminus\{2\}$ we have 
\begin{equation}
\Gamma_{2,i}+a_{n}\Gamma_{2,n}\Gamma_{n,i}=0.
\end{equation}
Taking $l=a_{n}\Gamma_{2,n}$ we get
\begin{equation}
\left\langle X_{2}Z_{2}^{l\Gamma_{2,n}}Z_{n}^{\Gamma_{2,n}}X_{n}^{l}\right\rangle_{\mathcal{I}_{1}^{3}}=\left\langle g_{2}g_{n}^{l}\right\rangle_{\mathcal{I}_{1}^{3}}.
\end{equation}
Now, we can use Fact \ref{app:fact_stab} to construct an implication
\begin{equation}
\left\langle g_{1}^{m}g_{2}\right\rangle_{\mathcal{I}_{1}^{3}}=\left\langle g_{1}^{-m}g_{n}^{l}\right\rangle_{\mathcal{I}_{1}^{3}}=1 \quad \Rightarrow \quad \left\langle g_{2}g_{n}^{l}\right\rangle_{\mathcal{I}_{1}^{3}}=1,
\end{equation}
where we take $m\in \mathbb{Z}_{d}$ such that
\begin{equation}
m=- \Gamma_{2,n}\Gamma_{1,n}^{-1}.
\end{equation}
It is easy to show that the $g_{1}^{m}g_{2}$ subnetwork of $\mathcal{I}_{1}^{3}$ does not contain party $n$ and that the $g_{1}^{-m}g_{n}^{l}$ subnetwork of $\mathcal{I}_{1}^{3}$ does not contain party $2$. Therefore, both of these subnetworks are connected, and so by tracing out all parties on which the respective operators act trivially we can conclude that
\begin{equation}
\left\langle g_{1}^{m}g_{2}\right\rangle_{\mathcal{I}_{1}^{3}}=\left\langle g_{1}^{m}g_{2}\right\rangle_{\mathcal{O}}=1,
\end{equation}
\begin{equation}
\left\langle g_{1}^{-m}g_{n}^{l}\right\rangle_{\mathcal{I}_{1}^{3}}=\left\langle g_{1}^{-m}g_{n}^{l}\right\rangle_{\mathcal{O}}=1.    
\end{equation}
Consequently, both operators in \eqref{eq_ineq_part3} stabilize a state generated in $\mathcal{I}_{0}^{3}$ and so the inequality \eqref{eq_ineq_part3} is violated which ends the proof.
\end{proof}

\section{Supplementary Note 5: Fidelity analysis}\label{app_fid}
This section contains a proof of Theorem 2 stated in the main text which for completeness 
we restate below. 
\setcounter{thm}{1}

\begin{thm}\label{app:thm_fid}
Let us consider a state $\rho$ and a graph state $\ket{G}$ that cannot be generated in an LOSR quantum network with bipartite sources. Moreover, let $q=|\mathcal{N}_{1}\cap \mathcal{N}_{2}|+1$ for graphs $G$ that fulfil \eqref{eq_N2_Nn} and $q=1$ in other cases. If
\begin{equation}
F(\rho, \ket{G}) >1- \frac{1}{8}\left( \beta^{2}+2\gamma - \beta \sqrt{\beta^{2}+4\gamma} \right),
\end{equation}
where $\beta=2q-1$ and $\gamma = (d-\sqrt{d})/(d-1)$, then $\rho$ cannot be generated in an LOSR quantum network with bipartite sources.
\end{thm}

\begin{proof}
Recall the definition of the fidelity $F(\rho,\sigma)$ between two states, where one of the states is a pure state $\sigma = \dyad{\psi}$:
\begin{align}
    F(\rho,\dyad{\psi}) = \expval{\rho}{\psi}.
\end{align}
Rewriting a state $\rho$ in an orthogonal basis that contains $\ket{\psi}$ and diagonalising the subspace orthogonal to $\ket{\psi}$, one gets
\begin{align}\label{eq:rho_decomp}
\begin{split}
  \rho &= (1-\delta) \dyad{\psi} + \sum_{i=1}^d c_i \dyad{f_i} \\
  &+ \sum_{j=1}^d (a_{0,j} \dyad{\psi}{f_j} + a_{0,j}^{*} \dyad{f_{j}}{\psi}),
\end{split}
\end{align}
where $\sum_{i=1}^d c_i = \delta$, so that $\tr(\rho) = 1$. Note, that this is always possible, since $\rho$ is Hermitian and thus so is every principal submatrix of $\rho$. This then yields
\begin{align}
    F(\rho, \dyad{\psi}) = 1-\delta,
\end{align}
and so we can express every state with fidelity $1-\delta$ with respect to a graph state $\ket{\psi}=\ket{G}$ in the form of Eq.~\eqref{eq:rho_decomp}.

In order to find our desired bound on the fidelity of $\rho$ we will need two facts. First, given a stabilizing operator $s$ of the graph state $\ket{G}$, we have
\begin{align}
\begin{split} \label{eq:expval_bound}
    \abs{\expval{\mathbb{1} - s}_{\rho}} &= \Big| 1- (1-\delta) \expval{s}_{\ket{G}}  \\
    &- \sum_{i}\tr( c_i s \dyad{f_i})\\
    &-  \sum_{j}\tr( a_{0,j} s\dyad{G}{f_i}  +  a_{0,j}^{*} \dyad{f_i}{G}s )   \Big| \\
    &\leq \delta +\sum_i c_i  = 2\delta.
\end{split}
\end{align}

Second, given two unitary operators $s_{1}$ and $s_{2}$ such that $\abs{\expval{\mathbb{1} - s_1}} \leqslant \mu, \abs{\expval{\mathbb{1} - s_2}} \leqslant \nu$, and $\mu\leqslant \nu$ the following holds true:
\begin{align}
\begin{split}\label{eq:propagation}
    &\abs{\expval{\mathbb{1}-s_1 s_2}} = \abs{\expval{(\mathbb{1} - s_1)s_2 + \mathbb{1} - s_2}}\\
    &\leq \abs{\sqrt{\expval{(\mathbb{1} - s_1)(\mathbb{1} - s_1^\dagger)}\expval{s_2^\dagger s_2}}} + \abs{\expval{\mathbb{1} - s_2}}\\
    &\leq \abs{\sqrt{\expval{2\mathbb{1} - s_1 - s_1^\dagger}}} + \nu\\
    &\leq \sqrt{2\mu} + \nu,
\end{split}
\end{align}
where we have used the triangle inequality and the Cauchy-Schwarz inequality. 

The next step is to assess how such a deviation from a unit expected value propagates in the reasoning of proof of Theorem \ref{app:thm_bip}. As the proof of Theorem \ref{app:thm_bip} is divided into three parts, we have to consider them separately.

We begin with \textit{Part 1} of the proof of Theorem \ref{app:thm_bip}. Unlike in the proof of Theorem \ref{app:thm_bip}, here we are forced to start from the "end" and consider the inflation $\mathcal{I}_{q}^{1}$ first. From Eq. \eqref{eq:expval_bound} it follows that
\begin{align}
|\langle \mathbb{1}- g_{1} \rangle_{\mathcal{I}_{q}^{1}}| \leqslant 2\delta, \qquad |\langle \mathbb{1}- g_{i}^{l} \rangle_{\mathcal{I}_{q}^{1}}| \leqslant 2 \delta,
\end{align}
where we assume that Eq. \eqref{eq_gi_1} holds true for $i\in T_{q-1}^{1}$. The case when Eq. \eqref{eq_gi_2} is true gives the same result, and so for simplicity we will only focus on  Eq. \eqref{eq_gi_1}.

From Eq. \eqref{eq:propagation} it follows that
\begin{equation}
|\langle \mathbb{1}- g_{1}g_{i}^{l} \rangle_{\mathcal{I}_{q}^{1}}| \leqslant 2\sqrt{\delta} + 2 \delta.
\end{equation}
From the proof of Theorem \ref{app:thm_bip} we have that 
\begin{equation}
\langle \mathbb{1}- g_{1}g_{i}^{l} \rangle_{\mathcal{I}_{q}^{1}} = \langle \mathbb{1}- g_{1}g_{i}^{l} \rangle_{\mathcal{I}_{q-1}^{1}},
\end{equation}
and so we also have
\begin{equation}
|\langle \mathbb{1}- g_{1}g_{i}^{l} \rangle_{\mathcal{I}_{q-1}^{1}}| \leqslant 2\sqrt{\delta} + 2 \delta.
\end{equation}
For inflation $\mathcal{I}_{q-1}^{1}$ we also have
\begin{equation}
|\langle \mathbb{1}- g_{i}^{-l} \rangle_{\mathcal{I}_{q-1}^{1}}| \leqslant 2 \delta.
\end{equation}
We can again use Eq. \eqref{eq:propagation} for $s_{1}=g_{i}^{-l}$ and $s_{2}=g_{1}g_{i}^{l}$ which yields
\begin{equation}
|\langle \mathbb{1}- g_{1} \rangle_{\mathcal{I}_{q-1}^{1}}| \leqslant 4\sqrt{\delta} + 2 \delta.
\end{equation}
Applying this procedure $q-1$ times, where $q=|\mathcal{N}_{1}\cap \mathcal{N}_{2}|+1$, gives us 
\begin{align}
    \abs{\expval{\mathbb{1} - \tilde{g}_{1}}_{\mathcal{I}_0^1}} &\leqslant 4(q-1)\sqrt{\delta} + 2\delta,
\end{align}
where $\tilde{g}_{1}=X_1 Z_{2'}Z_{N_1\setminus\{2\}}$. As this proof can be performed for any power $k\in\{1,\dots,d-1\}$ we actually have 
\begin{align}
    \abs{\expval{\mathbb{1} - \tilde{g}_{1}^{k}}_{\mathcal{I}_0^1}} &\leqslant 4(q-1)\sqrt{\delta} + 2\delta.
\end{align}
Similarly, one can show that 
\begin{equation}
\abs{\expval{\mathbb{1} - g_2^{k}}_{\mathcal{I}_0^1}} \leqslant 2\sqrt{\delta} +2\delta.
\end{equation}

\begin{figure}
    \centering
    \includegraphics[width=0.25\textwidth]{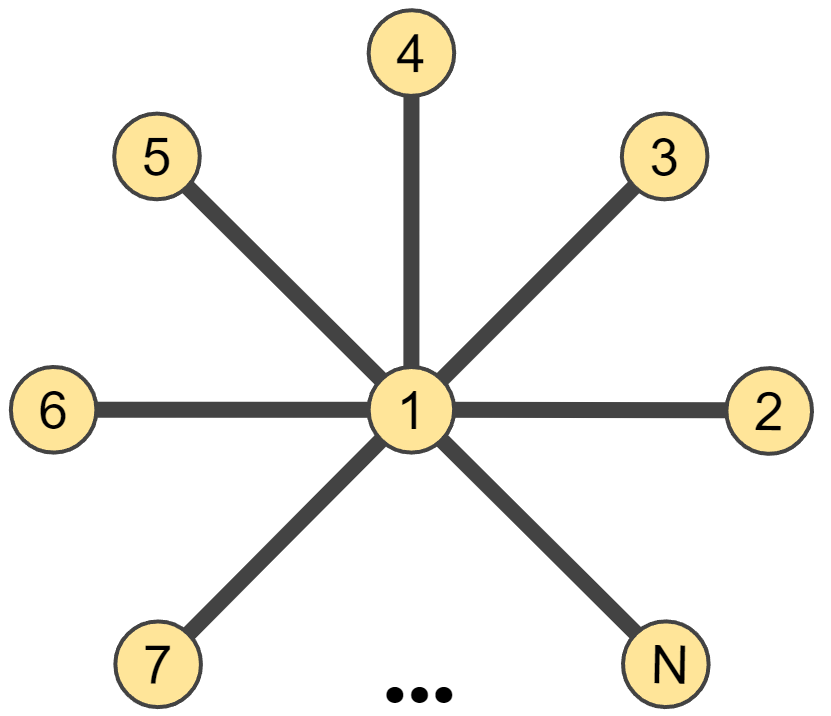}
    \caption{\textbf{A star graph with the central node $1$}. The edges can have any weight $\Gamma_{1,j}\in \{1,\ldots,d-1\}$. This graph represents a graph state with a relatively good fidelity bound.}
    \label{fig:star_graph}
\end{figure}

We can now use the above inequalities to find for which values of $\delta$ one still breaks the bounds of Lemma \ref{app:lem_ineq}: 
\begin{align}
\begin{split}
    2(d-1) - &\abs{ \sum_{k=1}^{d-1} \expval{\tilde{g}_1^k}_{\mathcal{I}_0^1} + \expval{g_2^k}_{\mathcal{I}_0^1} }\\
    &\leq \abs{ \sum_{k=1}^{d-1} \expval{\mathbb{1} - \tilde{g}_1^k}_{\mathcal{I}_0^1} + \expval{\mathbb{1} - g_2^k}_{\mathcal{I}_0^1} }\\
    &\leq 2(d-1) \left[ (2q-1)\sqrt{\delta} + 2\delta \right],
\end{split}
\end{align}
so that
\begin{multline}
    \abs{\sum_{k=1}^{d-1} \expval{\tilde{g}_1^k}_{\mathcal{I}_0^1} + \expval{g_2^{k}}_{\mathcal{I}_0^1}} \\
    \geq 2(d-1)\left[1-(2q-1)\sqrt{\delta} - 2\delta\right].
\end{multline}
Notice, that the $k=0$ terms in the inequality of Lemma~\ref{app:lem_ineq} always give exactly one which can be subtracted from both sides of an inequality. Hence, the bound of $d + \sqrt{d} - 2$ is violated if
\begin{align}\label{eq:fid_Ad}
 (2q-1)\sqrt{\delta}+2\delta < \frac{d-\sqrt{d}}{2(d-1)}.
\end{align}
Solving this inequality yields
\begin{equation}\label{eq_bound_delta}
\delta < \frac{1}{8}\left( \beta^{2}+2\gamma - \beta \sqrt{\beta^{2}+4\gamma} \right),
\end{equation}
where $\beta=2q-1$ and $\gamma = (d-\sqrt{d})/(d-1)$. 

As for \textit{Part 2} and \textit{Part 3} of the proof of Theorem \ref{app:thm_bip}, the derivation of the fidelity bound is relatively simple as compared to the case of \textit{Part 1}, and so we will only show the result. For both \textit{Part 2} and \textit{Part 3} the inequality is violated if
\begin{equation}
(\sqrt{\delta} + 2\delta) < \frac{d - \sqrt{d}}{2(d-1)},
\end{equation}
which yields
\begin{equation}\label{eq_bound_delta_op}
\delta < \frac{1}{8}(1+ 2 \gamma - \sqrt{1+4 \gamma}),
\end{equation}
where $\gamma = (d-\sqrt{d})/(d-1)$. Notice, that this bound corresponds to the bound \eqref{eq_bound_delta} for $q=1$, and so the case where $\mathcal{N}_{1}\cap \mathcal{N}_{2} = \emptyset$. Coincidentally, the left hand side of Eq. \eqref{eq_bound_delta} decreases with increasing $q$, and so Eq. \eqref{eq_bound_delta_op} is also an optimal case of Eq. \eqref{eq_bound_delta}.
\end{proof}

One of the examples of a graph state bounded by Eq. \eqref{eq_bound_delta_op} is a star graph state presented on Figure \ref{fig:star_graph}. Since we have $\mathcal{N}_1 \cap \mathcal{N}_2 = \emptyset$, then clearly this graph corresponds to the first case of Lemma \ref{app:lem_graph} for which $q=1$. Substituting $d=3$ yields $\delta < 0.0484$, while for $d\to \infty$ one gets $\delta < (3-\sqrt{5})/8 \approx 0.0955$.


%\bibliographystyle{naturemag}
%\bibliography{networks2}

\end{document}